\documentclass[final]{siamltex}

\usepackage{cite} 
\usepackage{url}  
\usepackage[utf8]{inputenc} 
 \usepackage{color}
 \usepackage[ruled,vlined,english,linesnumbered]{algorithm2e}
\usepackage{amssymb}
\usepackage{multirow}
\usepackage{float}
\usepackage{amsmath}
\usepackage{graphicx}

\newcommand{\NN}{\mathbb{N}}

\newcommand{\cP}{\mathcal{P}}

\newcommand{\cT}{\mathcal{T}}
\newcommand{\cS}{\mathcal{S}}
\newcommand{\pp}{\{1,\ldots ,p\}}

\newtheorem{observation}{Observation}

\newcommand{\leo}{}
\newcommand{\cs}{}
\newcommand{\steven}{}
\newcommand{\cyan}{}
\newcommand{\snew}{}
\newcommand{\sbsteven}{}

\newif\ifcomment\commentfalse

\def\commentOFF{\commentfalse}

\long\outer\def\bc#1\ec{\ifcomment \sloppy  {\textcolor{red}{#1}
 \fi }}

\long\outer\def\bg#1\eg{\ifcomment \sloppy  {\textcolor{green}{#1}
 \fi }}

\long\outer\def\bo#1\eo{\ifcomment \sloppy  {\textcolor{yellow}{#1}
 \fi }}

\long\outer\def\BC#1\EC{{\ifcomment \sloppy \par \textcolor{blue}{\#  \dotfill
{\textsc{#1}} \dotfill \#} \par \fi }}

\long\outer\def\BCF#1\ECF{{\ifcomment \sloppy \par \textcolor{red}{\#  \dotfill
{\textit{Frank: #1}} \dotfill \#} \par \fi }}

\long\outer\def\BCL#1\ECL{{\ifcomment \sloppy \par \textcolor{green}{\#  \dotfill
{\textsc{Leen: #1}} \dotfill \#} \par \fi }}

\commentOFF

\title{On Computing the Maximum Parsimony Score of a Phylogenetic Network}

\author{Mareike Fischer\thanks{Ernst-Moritz-Arndt University Greifswald, Department of Mathematics and Computer Science, Walther-Rathenau-Str. 47, 17487 Greifswald, Germany, \tt{email@mareikefischer.de}}. \and Leo van Iersel\thanks{Centrum Wiskunde \& Informatica (CWI), P.O. Box 94079, 1090 GB Amsterdam, The Netherlands, {\tt l.j.j.v.iersel@gmail.com}. Leo van Iersel was funded by a Veni grant of the Netherlands Organisation for Scientific Research (NWO).}
 \and Steven Kelk\thanks{Department of Knowledge Engineering (DKE), Maastricht University, P.O. Box 616, 6200 MD Maastricht, The Netherlands, \tt{steven.kelk@maastrichtuniversity.nl}.}
 \and Celine Scornavacca\thanks{Institut des Sciences de l'Evolution (ISEM, UMR 5554 CNRS), Universit\'e Montpellier~II, Place E. Bataillon - CC 064 - 34095 Montpellier Cedex 5, France, {\tt celine.scornavacca@univ-montp2.fr}. Celine Scornavacca was partially supported by the ANCESTROME project ANR-10-IABI-0-01.}
}

\begin{document}
  
\maketitle

\begin{abstract}
Phylogenetic networks are used to display the relationship among different species whose evolution is not treelike, which is the case, for instance, in the presence of hybridization events or horizontal gene transfers. Tree inference methods such as Maximum Parsimony need to be modified in order to be applicable to networks. In this paper, we discuss two different definitions of Maximum Parsimony on networks, ``hardwired'' and ``softwired'', and examine the complexity of computing them  given a network topology and a character. By exploiting a link with the problem \textsc{Multicut}, we show that computing the hardwired parsimony score for 2-state characters is polynomial-time solvable, while for characters with more states this problem becomes NP-hard but is still approximable and fixed parameter tractable in the parsimony score. On the other hand we show that, for the softwired definition, obtaining even weak approximation guarantees is already difficult for binary characters and restricted network topologies, and fixed-parameter tractable algorithms in the parsimony score are unlikely. On the positive side we show that computing the softwired parsimony score is fixed-parameter tractable in the level of the network, a natural parameter describing how tangled reticulate activity is in the network. Finally, we show that both the hardwired and softwired parsimony score can be computed efficiently using Integer Linear Programming. The software has been made freely available.\\
\end{abstract}



\pagestyle{myheadings}
\thispagestyle{plain}
\markboth{Fischer, Van Iersel, Kelk, Scornavacca}{On Computing the Maximum Parsimony Score of a Phylogenetic Network}

\section{Introduction} 
In phylogenetics, graphs are used to describe the relationships among different species. Traditionally, these graphs are trees, and biologists aim at reconstructing the so-called `tree of life', i.e. the tree of all living species \cite{treeoflife}. However, trees cannot display reticulation events such as hybridizations or horizontal gene transfers, which are known to play an important role in the evolution of certain species \cite{arnold_hybrid, koonin_HGT, bogart_hybrid, mcdaniel_HGT}. In such cases considering phylogenetic networks rather than trees is potentially more adequate, where in its broadest sense a phylogenetic network can simply be thought of as a graph (directed or undirected) with its leaves labelled by species \cite{HusonRuppScornavacca10,davidbook,Nakhleh2009ProbSolv}. \leo{Phylogenetic networks can also be useful in the absence of reticulate events, since such networks can represent uncertainty in the true (tree-shaped) phylogeny. Hence,} tree reconstruction methods, i.e. the methods used to infer the best tree from e.g. DNA or protein data, need to be adapted to networks. 

One of the most famous tree reconstruction methods is Maximum Parsimony \cite{EdwardsCavalli-Sforza64}. While this method has been shown to have drawbacks like statistical inconsistency in the so-called `Felsenstein zone' \cite{felsenstein_1978}, it is still widely used mainly due to its simplicity: Maximum Parsimony does not depend on a phylogenetic model and works in a purely combinatorial way. Moreover, for a given tree the optimal parsimony score can be found in polynomial time using the well-known Fitch algorithm~\cite{fitch_1971} or the Sankoff one~\cite{sankoff_75}. This problem of finding the optimal parsimony score for a given tree is often referred to as the ``small parsimony" problem. The ``big parsimony" problem, on the other hand, aims at finding the most parsimonious tree amongst all possible trees -- and this problem has been proven to be NP-hard. Note that the latter problem is a close relative of the classical \textsc{Steiner Tree} problem \cite{foulds_1982,AlonApproxMP}.

Recent studies have introduced extensions of the tree-based parsimony concept to phylogenetic networks \cite{nguyen2007reconstructing, jin2009parsimony,kannan2012maximum} and a biological case-study was presented by \cite{MPcaseStudy2007}. Basically, Maximum Parsimony on networks can be viewed in two ways: If one thinks of evolution as a tree-like process (but maybe with different trees for different \leo{parts of the genome}, all of which are represented by a single network), one can define the parsimony score of a character on a network as the score of the best tree inside the network. The other way of looking at Maximum Parsimony on networks is just the same as the Fitch algorithm's view on trees: One can try to find the assignment of states to internal nodes of the network such that the total number of edges that connect nodes in different states is minimized. While the first concept may be regarded as more biologically motivated, the second one is in a mathematical sense the natural extension of the parsimony concept to networks. Both concepts of parsimony on networks are considered in this manuscript, and we formally introduce them in Section~\ref{prelim_sec} as softwired and hardwired parsimony, respectively.

Given an alignment (e.g. DNA) and a criterion like Maximum Parsimony, several questions come to mind: How hard is it to calculate the parsimony score (both in the hardwired and softwired sense) for a given network (``small parsimony" problem)? How hard is it to find the best network (``big parsimony" problem)? In the present paper we \leo{consider only} the
``small parsimony" problem. \leo{Moreover, we focus on computing the parsimony score of a single character, since the parsimony score of an alignment can be computed by summing up the parsimony scores of the individual charactes.} 

\leo{Kannan and Wheeler introduced the hardwired parsimony score for networks and conjectured that it would be NP-hard to compute~\cite{kannan2012maximum}.} We show in Section~\ref{hardwired_sec} that this problem is indeed NP-hard and APX-hard (Corollary \ref{hardwired_NP_cor}) whenever characters employing more than two states are used, but we also show that it is polynomial-time solvable for binary characters (Corollary~\ref{hardwired_bin_cor}). We also analyse the behaviour of \leo{the algorithm in~\cite{kannan2012maximum}, which we call the \textsc{ExtendedFitch} algorithm}, showing that it does not compute the hardwired parsimony score \leo{optimally and that it does not approximate} the softwired parsimony score well.
 
In Section~\ref{softwired_sec}, \leo{we consider the complexity of computing the softwired parsimony score. Previously, this problem was shown to be NP-hard and APX-hard but only for nonbinary networks with outdegree at most~20~\cite{jin2009parsimony}. We show in Theorem~\ref{thm:rootedhard} that the softwired parsimony problem is NP-hard even for binary networks (and binary characters) and we additionally show that NP-hardness cannot be overcome by considering only so-called binary tree-child time-consistent networks (Theorem \ref{softwired_treechild_thm}). Moreover, we show a much stronger inapproximability than the APX-hardness shown in~\cite{jin2009parsimony}. We show that, for any constant $\epsilon > 0$, an approximation factor of $|X|^{1-\epsilon}$ is not possible in polynomial time, unless P~$=$~NP, where $|X|$ denotes the number of species under investigation. Moreover,} this holds even for tree-child time-consistent networks (Theorem \ref{thm:sat}). Our inapproximability result shows that a trivial approximation factor of $|X|$ is in a certain sense the best that is possible in polynomial time.  
\leo{For binary networks we show a slightly weaker inapproximability threshold:  $|X|^{\frac{1}{3}-\epsilon}$ (Theorem \ref{thm:sat2}).} We note that the hardness results \leo{in}~\cite{nguyen2007reconstructing} are not directly comparable \leo{to our results}, since they adopt the \emph{recombination network} model of phylogenetic networks (see e.g. \cite{HusonRuppScornavacca10,twotrees} and also the discussion \leo{in} \cite{jin2009parsimony}).

While \leo{we show that} the hardwired parsimony score \leo{is} fixed parameter tractable with the parsimony score as parameter (Corollary~\ref{hardwired_FPT_cor}), \leo{we show that} the softwired parsimony score is not, unless P~$=$~NP. Indeed, we show \leo{(in Corollary~\ref{cor:notfpt})} that it is even NP-hard to determine whether the softwired parsimony score is \leo{equal to one (see \cite{Flum2006} and \cite{niedermeier2006} for an introduction to fixed parameter tractability). On the positive side,} we show in Section~\ref{sec:FPT} (in Theorem~\ref{softwired_levelFPT_thm}) that the softwired parsimony score is fixed parameter tractable in the \emph{level} of the network, \leo{a parameter describing} the maximum amount of reticulate activity in a biconnected component of the network (see e.g. \cite{KelkScornavacca2011} and \cite{elusiveness} for an overview). Moreover, in Section~\ref{ilp_sec} we present an Integer Linear Program to calculate both the softwired and hardwired parsimony score of a given character (or, more generally, multiple sequence alignment) on a phylogenetic network and give a preliminary analysis of its performance. This is the first practical exact method for computation of parsimony on medium to large networks, supplementing the heuristics given by \cite{jin2009parsimony} and \cite{kannan2012maximum}. An implementation of this program is freely available~\cite{MPNet}.

Finally, in Section~\ref{conclusion_sec}, we summarize our results and state some open problems for future research.

\section{Preliminaries}\label{prelim_sec}
Let~$X$ be a finite set. An \emph{unrooted phylogenetic network} on~$X$ is a connected, undirected graph that has no degree-2 nodes and that has its degree-1 nodes (the leaves) bijectively labelled by the elements of~$X$. A \emph{rooted phylogenetic network} on~$X$ is a directed acyclic graph that has a single indegree-0 node (the root), no indegree-1 outdegree-1 nodes, and its outdegree-0 nodes (the leaves) bijectively labelled by the elements of~$X$. We identify each leaf with its label. The indegree of a node~$v$ of a rooted phylogenetic network is denoted~$\delta^-(v)$ and~$v$ is said to be a \emph{reticulation node} if~$\delta^-(v)\geq 2$. An edge~$(u,v)$ is called a \emph{reticulation edge} if~$v$ is a reticulation node and it is called a \emph{tree edge} otherwise. A proper subset $C\subset X$ is referred to as \snew{a} \emph{cluster} of $X$.

When we refer to a \emph{phylogenetic network}, it can be either rooted or unrooted. We use $V(N)$ and $E(N)$ to denote, respectively, the node and edge set of a phylogenetic network~$N$. To simplify notation, we use the notation $(u,v)$ for a directed as well as for an undirected edge between~$u$ and~$v$. A phylogenetic network is \emph{binary} if each node has total degree at most~3 and (in case of a rooted network) the root has outdegree~2 and all reticulations have outdegree~1.

The \emph{reticulation number} of a phylogenetic network~$N$ can be defined as $|E(N)|-|V(N)|+1$. Hence, the reticulation number of a rooted binary network is simply the number of reticulation nodes. A \emph{rooted phylogenetic tree} is a rooted phylogenetic network with no reticulation nodes, i.e. with reticulation number~0. A \emph{biconnected component} of a phylogenetic network is a maximal biconnected subgraph \leo{(i.e. a biconnected subgraph that is not contained in a larger biconnected subgraph)}. A phylogenetic network is said to be a \emph{level}-$k$ network if each biconnected component has reticulation number at most~$k$, \cs{and at least one biconnected component has reticulation number exactly~$k$}.

A rooted phylogenetic network~$N$ is said to be \emph{tree-child} if each non-leaf node has a child that is not a reticulation and~$N$ is said to be \emph{time-consistent} if there exists a ``time-stamp'' function~$t:V(N)\rightarrow\NN$ such that for each edge~$(u,v)$ holds that $t(u)=t(v)$ if~$v$ is a reticulation and $t(u)<t(v)$ otherwise \cite{BSS06}.

If~$F$ is a finite set and~$p\in\NN$, then a $p$-state \emph{character} on~$F$ is a function from~$F$ to~$\pp$. A $p$-state character is \emph{binary} if~$p=2$. Let~$\alpha$ be a $p$-state character on~$X$ and~$N$ a phylogenetic network on~$X$. Then, a $p$-state character~$\tau$ on~$V(N)$ is an \emph{extension} of~$\alpha$ to~$V(N)$ if $\tau(x)=\alpha(x)$ for all~$x\in X$. Given a $p$-state character~$\tau$ on~$V(N)$ and an edge~$e=(u,v)$ of~$N$, the \emph{change} $c_\tau(e)$ on edge~$e$ w.r.t.~$\tau$ is defined as:
\[
c_\tau(e) = \left\lbrace \begin{array}{l}
                 0 \text{ if } \tau(u)=\tau(v)\\
                 1 \text{ if } \tau(u)\neq\tau(v).
               \end{array}\right.
\]

The \emph{hardwired parsimony score} of a phylogenetic network~$N$ and a $p$-state character~$\alpha$ on~$X$ can be defined as:
\[
PS_\text{hw}(N,\alpha) = \min_{\tau} \sum_{e\in E(N)} c_\tau(e),
\]

where the minimum is taken over all extensions~$\tau$ of~$\alpha$ to~$V(N)$. 

Now, consider a phylogenetic network~$N$ on~$X$ and a phylogenetic tree~$T$ on~$X$, where either both~$N$ and~$T$ are rooted or both are not. We say that~$T$ is \emph{displayed} by~$N$ if~$T$ can be obtained from a subgraph of~$N$ by suppressing non-root nodes with total degree~2. For a rooted phylogenetic network~$N$, a \emph{switching} of $N$ is obtained by, for each reticulation node, deleting all but one of its incoming edges \cite{KelkScornavacca2011}. We denote the set of switchings of $N$ by $\cS(N)$. It can easily be seen that~$T$ is displayed by~$N$ if and only if~$T$ can be obtained from a switching of~$N$ by deleting indegree-0 outdegree-1 nodes, deleting unlabelled outdegree-0 nodes and suppressing indegree-1 outdegree-1 nodes. Let~$\cT(N)$ denote the set of all phylogenetic trees on~$X$ that are displayed by~$N$. The \emph{softwired parsimony score} of a phylogenetic network~$N$ and a $p$-state character~$\alpha$ on~$X$ can be defined as:
\[
PS_\text{sw}(N,\alpha) = \min_{T\in\cT(N)} \min_{\tau} \sum_{e\in E(T)} c_\tau(e),
\]

where the second minimum is taken over all extensions~$\tau$ of~$\alpha$ to~$V(T)$.

Note that both the hardwired and softwired parsimony score can be used for rooted as well as for unrooted networks although the softwired parsimony score might seem more relevant for rooted networks and the hardwired parsimony score for unrooted ones.

It can easily be seen that, if~$N$ is a tree, $PS_\text{hw}(N,\alpha) = PS_\text{sw}(N,\alpha)$. However, we now show that, if~$N$ is a network, the difference between the two can be arbitrarily large.

\begin{figure}[H]
    \centering
    \includegraphics[scale=0.5]{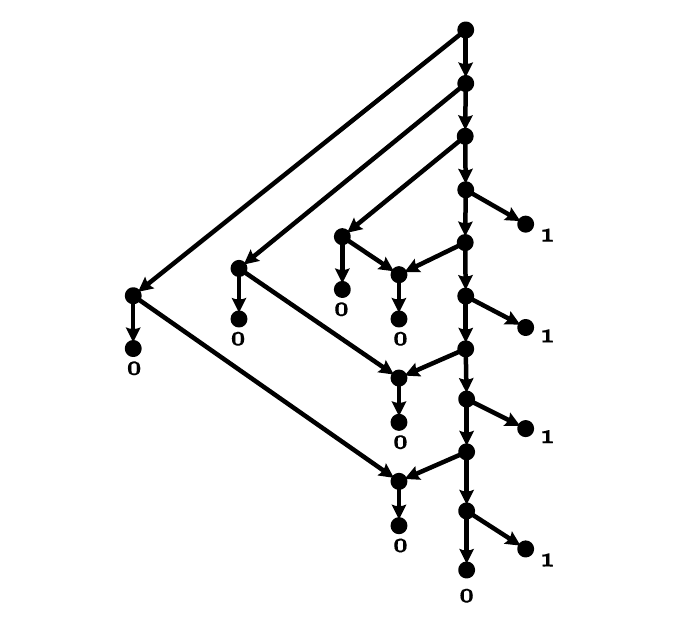}
    \caption{Example of a rooted phylogenetic network $N$ and binary character $\alpha$ for which the difference between $PS_\text{hw}(N,\alpha)$ and $PS_\text{sw}(N,\alpha)$ is arbitrarily large ($r-1$ if the network is extended to have $r$ reticulations and $n=3r+2$ leaves) and for which \textsc{ExtendedFitch} does not compute an $o(n)$-approximation of $PS_\text{sw}(N,\alpha)$.\label{fig:bad_fitch}}
\end{figure}

Figure~\ref{fig:bad_fitch} presents an example of a rooted binary phylogenetic network~$N$ and a binary character $\alpha$, where $PS_\text{sw}(N,\alpha) = 2$ regardless of the number of reticulation nodes in~$N$. This is due to the fact that all right-hand side parental edges of the reticulation nodes could be switched off, such that only one change from~0 to~1 would be required in the resulting tree on the edge just above all taxa labelled ``1'' and another change from~1 to~0 in the (0,1)-cherry. However, $PS_\text{hw}(N,\alpha)$ can be made arbitrarily large by extending the \snew{construction in the expected fashion}, as in this network we have $PS_\text{hw}(N,\alpha)=r+1$, where~$r$ denotes the number of reticulation nodes in~$N$. So the difference $PS_\text{hw}(N,\alpha)-PS_\text{sw}(N,\alpha)$ equals $r-1$, where~$r$ can be made arbitrarily large, which shows that $PS_\text{hw}(N,\alpha)$ is not an $o(n)$-approximation of $PS_\text{sw}(N,\alpha)$, where $n$ is the number of taxa. \snew{Note that the construction
shown in Figure \ref{fig:bad_fitch} is binary, tree-child and time-consistent}.

\subsection{Extended Fitch Algorithm}\label{subsec:fitch}
The hardwired parsimony score for rooted networks has been introduced by \cite{kannan2012maximum} who proposed a heuristic by extending the well-known Fitch algorithm for trees. We will call this algorithm \textsc{ExtendedFitch}. \leo{Its description can be found in Algorithm~3 and~4 of~\cite{kannan2012maximum}.} We show in this section that \textsc{ExtendedFitch} does not compute the hardwired parsimony score optimally and does not provide a good approximation for the softwired parsimony score.

The example from Figure~\ref{fig:bad_fitch} can be used again, in order to show that \textsc{ExtendedFitch} does not compute an $o(n)$-approximation of the softwired parsimony score. Indeed, \textsc{ExtendedFitch} gives all internal nodes character state~0, leading to a total score of $r+1 = (n+1)/3$ (which is in this case indeed equal to the hardwired parsimony score), while we already showed that $PS_\text{sw}(N,\alpha)=2$. Hence, the approximation ratio of \textsc{ExtendedFitch} is at least $(r+1)/2$ as a function of~$r$ and at least $(n+1)/6$ as a function of~$n$.

\steven{Note that Theorem \ref{thm:sat} is furthermore complexity-theoretic evidence that \textsc{ExtendedFitch}, a polynomial-time algorithm, cannot approximate the softwired parsimony score of a network well} (unless~P~=~NP).

\begin{figure}[H]
    \centering
   \includegraphics[scale=0.5]{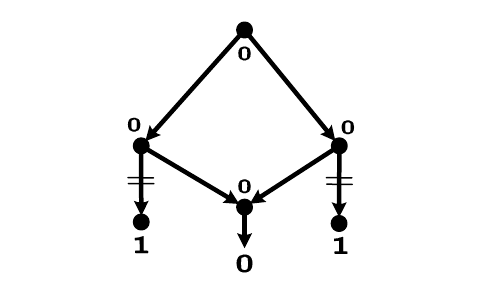} \includegraphics[scale=0.5]{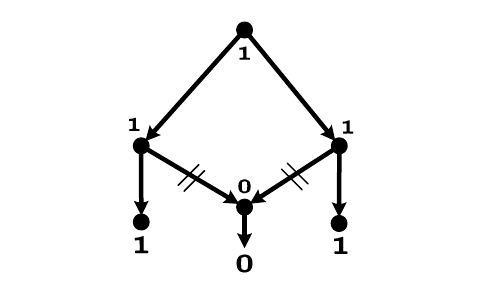}
    \caption{Example of a rooted phylogenetic network $N$ and binary character $\alpha$ for which \textsc{ExtendedFitch} does not provide the optimal parsimony score $PS_\text{hw}(N,\alpha)$. The small numbers refer to the two possible internal labellings as suggested by \textsc{ExtendedFitch}. Both require two changes on the marked edges. However, the optimal parsimony score is~1: If all internal nodes are labelled~1, then only one change is needed on the edge from the reticulation node to the leaf labelled~0. \label{fig:extendedfitch}}
\end{figure}

Next we show that \textsc{ExtendedFitch} does not compute $PS_\text{hw}(N,\alpha)$ optimally, even if $\alpha$ is a binary character. Consider Figure~\ref{fig:extendedfitch}. This figure displays a rooted phylogenetic network~$N$ with three leaves, one of which is directly connected to the only reticulation node. \snew{The network is again binary, tree-child and time-consistent.} \textsc{ExtendedFitch} fixes the reticulation node to be in state~0, whereas all other internal nodes can be either~0 or~1. The two optimal solutions found by \textsc{ExtendedFitch} are illustrated by Figure~\ref{fig:extendedfitch}; they either uniformly set all internal nodes other than the reticulation node to state~0 or to state~1. Thus, the resulting score is~2, as either both pending edges leading to the leaves labelled~1 need a change or both edges leading to the reticulation node. The most parsimonious solution, however, would be to set all internal nodes -- including the reticulation node -- to state~1. This way, only one change on the edge from the reticulation node to its pending leaf would be required. Therefore, \textsc{ExtendedFitch} cannot be used to calculate the hardwired parsimony score of a character on a network.

\subsection{\cyan{A Comment on Model Differences}}
Our \snew{core} definition \snew{of phylogenetic network is slightly} less restricted than the definitions given by \cite{jin2009parsimony} and \cite{kannan2012maximum}. 
\snew{However, the strong hardness and inapproximability results we give in this article still hold under heavy topological and biological restrictions (degree restrictions, tree-child, time-consistent) that are often subsumed into the core definitions given in other articles.}
Moreover, an obvious advantage of our definition is that all the positive results in the article apply to the largest possible class of phylogenetic networks. 

\section{Computing the Hardwired Parsimony Score of a Phylogenetic Network}\label{hardwired_sec} 

Given an undirected graph~$G$ and a set~$\Gamma$ of nodes of~$G$ called \emph{terminals}, a \emph{multiterminal cut} of~$(G,\Gamma)$ is a subset~$E'$ of the edges of~$G$ such that each terminal is in a different connected component of the graph obtained from~$G$ by removing the edges of~$E'$. A \emph{minimum multiterminal cut} is a multiterminal cut of minimum size. The following theorem shows that computing the hardwired parsimony score of a phylogenetic network is at most as hard as \textsc{Multiterminal Cut}, the problem of finding a minimum multiterminal cut.

\begin{theorem}\label{thm:multicut}
Let~$N$ be a  phylogenetic network on~$X$ and~$\alpha$ a $p$-state character on~$X$. Let~$G$ be the graph obtained from~$N$ by merging all leaves~$x$ with~$\alpha(x)=i$ into a single node~$ \gamma_i$, for~$i=1,\ldots ,p$. Then, the size of a minimum multiterminal cut of $(G,\{ \gamma_1,\ldots, \gamma_k\})$ is equal to $PS_\text{hw}(N,\alpha)$.
\end{theorem}
\begin{proof}
First consider an extension~$\tau$ of~$\alpha$ to~$V(N)$ for which $PS_\text{hw}(N,\alpha) = \sum_{e\in E(N)} c_\tau(e)$ (i.e. an optimal extension). Let~$E'$ be the set of edges~$e$ with $c_\tau(e)=1$. Since $\tau( \gamma_i)=i$ for $i=1,\ldots,p$, any path from~$ \gamma_i$ to~$ \gamma_j$ with~$i\neq j$ contains at least one edge of~$E'$. Hence,~$E'$ is a multiterminal cut. Moreover, $PS_\text{hw}(N,\alpha) = \sum_{e\in E(N)} c_\tau(e) = |E'|$. Hence, $PS_\text{hw}(N,\alpha)$ is greater \snew{than} or equal to the size of a minimum multiterminal cut.

Now consider a minimum multiterminal cut~$E'$ of~$G$ and let~$G'$ be the result of removing the edges in~$E'$ from~$G$. We define an extension~$\tau$ of~$\alpha$ to~$V(N)$ as follows. First, we set $\tau(x)=\alpha(x)$, for all~$x\in X$. Then, for each node~$v$ that is in the same connected component of~$G'$ as~$ \gamma_i$, set $\tau(v)=i$. Finally, for each remaining node, set~$\tau(v)=p$. Then, each edge~$e\notin E'$ has $c_\tau(e)=0$. Consequently, each edge~$e$ with~$c_\tau(e)=1$ is in~$E'$. Hence, $PS_\text{hw}(N,\alpha) \leq \sum_{e\in E(N)} c_\tau(e) \leq |E'|$. It follows that $PS_\text{hw}(N,\alpha)$ is less or equal to the size of a minimum multiterminal cut, which concludes the proof.
\end{proof}
\begin{corollary} \label{hardwired_bin_cor}
Computing the hardwired parsimony score of a   phylogenetic network and a binary character is polynomial-time solvable.
\end{corollary}
\begin{proof}
This follows directly from Theorem~\ref{thm:multicut} because, in the case of two terminals, \textsc{Multiterminal Cut} becomes the classical minimum $s-t$-cut problem, which is polynomial-time solvable.
\end{proof}
\begin{corollary}\label{hardwired_NP_cor}
Computing the hardwired parsimony score of a   phylogenetic network and a $p$-state character, for~$p\geq 3$, is NP-hard and APX-hard.
\end{corollary}
\begin{proof}
We reduce from \textsc{Multiterminal Cut}, which is NP-hard and APX-hard \snew{for three or more terminals}~\cite{Dahlhaus94thecomplexity}.

Let~$G$ be an undirected graph and~$\Gamma=\{ \gamma_1,\ldots , \gamma_k\}$ a set of terminals. \snew{Note that feasible solutions to \textsc{Multiterminal Cut} must contain all edges between adjacent terminals. For this reason we begin by removing such edges from $G$. Next we repeatedly delete all degree-1 nodes that are not terminals, until no such nodes are left, because the edges adjacent to such nodes cannot contribute to a multiterminal cut.}

 We construct a finite set~$X$ and a \snew{$k$}-state character~$\alpha$ on~$X$ as follows. For each terminal~$ \gamma_i$, and for each node~$v_i^j$ adjacent to~$ \gamma_i$ in~$G$, put an element~$x_i^j$ in~$X$ and set~$\alpha(x_i^j)=i$. Now we construct a new graph~$N$ from~$G$ by deleting each~$ \gamma_i$ and adding a leaf labelled~$x_i^j$ with an edge $(v_i^j,x_i^j)$, for each~$x_i^j\in X$.
\snew{Now, $N$ might contain degree-2 nodes, which are not permitted in our definition of phylogenetic network, but
as we explain in Appendix \ref{app:transf}, there is a simple transformation that removes such nodes without altering the
hardwired parsimony score or the cut properties of the graph. We apply this transformation to $N$ if necessary. Suppose then} that the resulting graph~$N$ is connected, and hence an unrooted phylogenetic network on~$X$. Then it follows from Theorem~\ref{thm:multicut} that the size of a minimum multiterminal cut of $(G,\Gamma)$ is equal to $PS_\text{hw}(N,\alpha)$.

Now suppose that~$N$ is not connected. Observe that the proof of Theorem~\ref{thm:multicut} still holds if~$N$ is not connected. Moreover, computing the hardwired parsimony score of a connected unrooted phylogenetic network is at least as hard as computing the hardwired parsimony score of a not-necessarily connected phylogenetic network, because we can sum the parsimony scores of the connected components. This reduction is clearly approximation-preserving. Finally we note that computing the hardwired parsimony score of a rooted phylogenetic network is just as hard as computing this score of an unrooted phylogenetic network because the hardwired parsimony score does not depend on the orientation of the edges.
\end{proof}
\begin{corollary} \label{hardwired_FPT_cor}
Computing the hardwired parsimony score of a   phylogenetic network and a $p$-state character is fixed-parameter tractable (FPT) in the parsimony score. Moreover, there exists a polynomial-time 1.3438-approximation for all~$p$ and a $\frac{12}{11}$-approximation for $p=3$.
\end{corollary}
\begin{proof}
The approximation results follow from the corresponding results on minimum multiterminal cut~\cite{KargerSTOC09} by Theorem~\ref{thm:multicut}.

For the fixed-parameter tractability, we use the corresponding result on the problem \textsc{Multicut}, which is defined as follows. Given a graph~$G$ and~$q$ terminal pairs $(\gamma'_1, \gamma_1),\ldots ,(\gamma'_q, \gamma_q)$, find a minimum-size subset~$E'$ of the edges of~$G$ such that there is no path from~$\gamma'_i$ to~$ \gamma_i$ for,~$i=1,\ldots ,q$, in the graph obtained from~$G$ by removing the edges of~$E'$. Clearly, \textsc{Multiterminal Cut} can be reduced to \textsc{Multicut} by creating a terminal pair $(\gamma,\gamma')$ for each combination of two terminals $\gamma,\gamma'\in \Gamma$ with~$\gamma\neq \gamma'$. Hence, since \textsc{Multicut} is fixed-parameter tractable in the size of the cut~\cite{MulticutFPT2,MulticutFPT1}, it follows by Theorem~\ref{thm:multicut} that computing the hardwired parsimony score of a phylogenetic network and a $p$-state character is fixed-parameter tractable in the parsimony score.
\end{proof}

\section{Complexity of Computing the Softwired Parsimony Score of a Rooted Phylogenetic Network}\label{softwired_sec}

\leo{In Section~\ref{subsec:exact}, we consider the complexity of computing the softwired parsimony score exactly. Subsequently, Section~\ref{subsec:approx} determines the complexity of approximating this score.}

\subsection{Complexity of Computing the Softwired Parsimony Score Exactly}\label{subsec:exact}

In the following, we show that computing the softwired parsimony score of a binary character on a binary rooted phylogenetic network is NP-hard. We reduce from \textsc{Cluster Containment}, which is known to be NP-hard for general \leo{networks~\cite{clustercontainment,ISS2010b}.} However, in order to prove our result for binary networks, we first need to show that \textsc{Cluster Containment} is NP-hard for binary phylogenetic networks, too; \cyan{this intermediate result has to the best of our knowledge not appeared earlier in the literature}. We do this via \textsc{Edge Cluster Containment} and \textsc{Binary Edge Cluster Containment} as described in the following. Thus, we state the following questions and analyse their \leo{complexity.}

\textsc{(Binary) Edge Cluster Containment}\\
\emph{Instance: } A set $X$ of taxa, a rooted \snew{(binary)} phylogenetic network $N$ (with edge set $E$ and node set $V$) on $X$, a cluster $C \subset X$ and an edge $e=(u,v) \in E$. \\
\emph{Question:} Is there a \snew{rooted phylogenetic} tree $T$ on $X$ displayed by $N$ containing the node $v$  such \snew{that} the taxa descending from $v$ in $T$ are precisely the taxa in $C$?\par

\leo{If} the answer is yes to the question above, we say that $e$ represents $C$. We denote by $\mathcal{C}(N)$ the set of clusters represented by edges in $E$.

\textsc{(Binary) Cluster Containment}\\
\emph{Instance: } A set $X$ of taxa, a rooted \snew{(binary)} phylogenetic network $N$ (with edge set $E$ and node set $V$) on $X$ and a cluster $C \subset X$. \\
\emph{Question:} Is there a \snew{rooted phylogenetic} tree $T$ on $X$ displayed by $N$ which contains an edge $e=(u,v)\in E$ such that the taxa descending from $v$ in $T$ are precisely the taxa in $C$?\par

The following \leo{observation follows directly by a simple (Turing) reduction from \textsc{Cluster Containment}.
\begin{observation}\label{obs:edgecluster}
\textsc{Edge Cluster Containment} is NP-hard.
\end{observation}}

Next we use the previous observation to prove the following \leo{small} result.

\begin{observation}\label{obs:edgeclusterbinary}
\textsc{Binary Edge Cluster Containment} is NP-hard.
\end{observation}
\begin{proof} We reduce from \textsc{Edge Cluster Containment}. Assume there is an algorithm $\mathcal{A}$ to decide \textsc{Binary Edge Cluster Containment} in polynomial time. Let $N$ be a phylogenetic network on $X$  containing an edge $e$. We want to know if $e$ represents a particular cluster $C$. Let $N^B$ be an arbitrary binary refinement of $N$. Note that $N^B$ contains all edges of $N$, and possibly some more (unless $N$ is already binary), \snew{in the sense
that edges in $N^B$ could be contracted to once again obtain $N$. Hence,} $e$ is contained in $N^B$, too. An example of a binary refinement of a \snew{nonbinary} network is depicted in Figure~\ref{fig:refinement}. So we can use $\mathcal{A}$ to decide if $e$ represents $C$ in $N^B$. Note that $e$ represents $C$ in $N^B$ if and only if $e$ represents $C$ in $N$, because it is easy to see that refining a network does not change the clusters pending on a particular edge. Therefore, this method would provide a polynomial-time algorithm to solve \textsc{Edge Cluster Containment}.
\end{proof}

\begin{figure}[H]
    \centering
    \includegraphics[scale=.6]{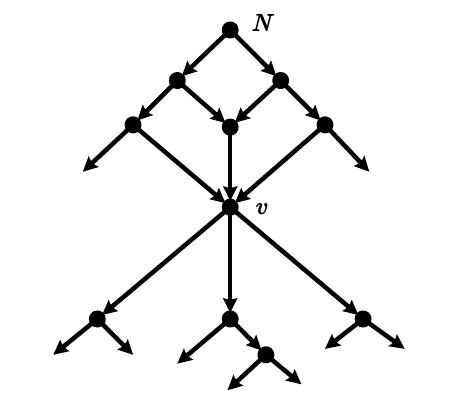} \includegraphics[scale=.6]{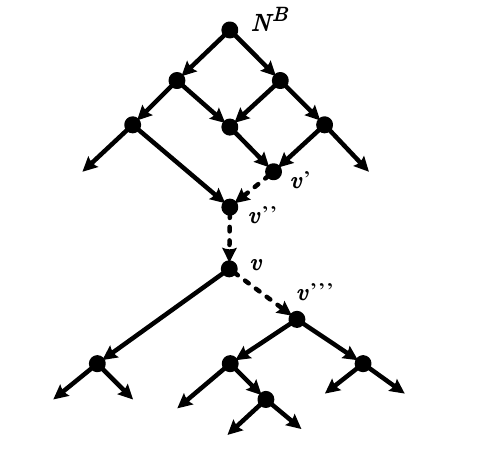}
    \caption{Illustration of a rooted phylogenetic network $N$ with a node $v$ of total degree 6 and a possible binary refinement $N^B$ of $N$, where three copies of $v$, namely $v'$, $v''$ and $v'''$, as well as three new edges (dashed lines) are inserted.\label{fig:refinement}}
\end{figure}

Now we are in a position to prove that \snew{\textsc{Binary Cluster Containment}} is NP-hard, which is the essential ingredient to our proof of Theorem~\ref{thm:rootedhard}.

\begin{lemma}\label{lem:clusterbinary}
\textsc{Binary Cluster Containment} is NP-hard.
\end{lemma}

\begin{proof} We reduce from \textsc{Binary Edge Cluster Containment}. Let $N^B$ be a rooted binary phylogenetic network  on $X$ and $C$ a cluster of $X$. Assume there is \snew{an} algorithm $\mathcal{A}$ to answer \textsc{Binary Cluster Containment} in polynomial time. Let $e=(v,u)$ be an edge in $N^B$. We add two new nodes $v_1$, $v_2$ \leo{to $N^B$ as follows}: Subdivide $e$ into three edges $e_1:=(v,v_1)$, $e_2:=(v_1,v_2)$ and $e_3:=(v_2,u)$. Now introduce two new edges $e_4:=(v_1,h_2)$ and $e_5:=(v_2,h_1)$, \snew{where} $h_1$ and $h_2$ are two new taxa. We call the resulting modified network \snew{$\tilde{N}^B$}. An example of this transformation  is depicted in Figure \ref{fig:edgesubdivision}. 

\begin{figure}[H]
    \centering
    \includegraphics[scale=0.7]{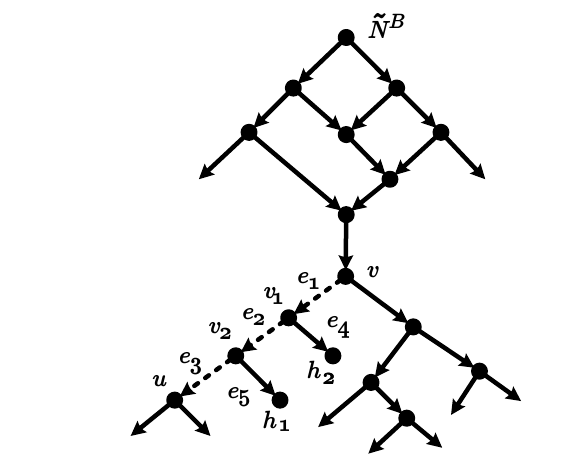} 
    \caption{Illustration of the modifications applied to $N^B$ as depicted by Figure \ref{fig:refinement}, resulting in the modified binary network \snew{$\tilde{N}^B$.}\label{fig:edgesubdivision}}
\end{figure}

Note that by construction, \snew{$\tilde{N}^B$} is binary. We now use algorithm $\mathcal{A}$ to decide in polynomial time if \snew{$\tilde{N}^B$} contains the cluster $C \cup h_1$. Note that this is the case if and only if \snew{$e_2$} in \snew{$\tilde{N}^B$} represents $C$, which, by construction, is the case if and only if $e$ represents $C$ in $N^B$. Therefore, this method would provide a polynomial-time algorithm to solve \textsc{Binary Edge Cluster Containment}.
\end{proof}

The following theorem was shown by \cite{jin2009parsimony} for nonbinary networks. The advantage of the proof given below is that it shows that the problem is even NP-hard for binary networks, \cyan{demonstrates a direct and insightful relationship between cluster containment and parsimony},
and \cyan{leads directly to the conclusion} that the problem is not even fixed-parameter tractable (unless ~P~=~NP).

\begin{theorem}\label{thm:rootedhard}
Computing the softwired parsimony score of a binary character on a binary rooted phylogenetic network is NP-hard.
\end{theorem}
\begin{proof}
We reduce from \textsc{Binary Cluster Containment}. Let $N$ be a rooted binary phylogenetic network on taxon set $X$ and let $C\subset X$ be a cluster. Then, by definition of $\mathcal{C}(N)$, $C$ is in $\mathcal{C}(N)$ if and only if there is a tree $T$ on $X$ displayed by $N$ with an edge $e=(u,v)$ such that the taxa descending from $v$ in $T$ are precisely the elements of $C$. This is the case if and only if $v$ is the root of a subtree of $T$ with leaf set $C$. Now assume that there is an algorithm $\mathcal{A}$ to compute the softwired parsimony score of a binary character on a rooted binary phylogenetic network in polynomial time. Then, we can solve \textsc{Binary Cluster Containment} by the following algorithm $\mathcal{\tilde{A}}$: 
\begin{enumerate} 
\item Introduce a modified version $\hat{N}$ of $N$ as follows: Add an additional taxon $z$ to $N$ and a new node $\hat{\rho}$ as well as the edges \snew{$(\hat{\rho},z)$ and $(\hat{\rho}, \rho)$}, where $\rho$ is the root of $N$. Thus, the taxon set $\hat{X}$ of $\hat{N}$ is $X \cup \{z\}$ and the root of $\hat{N}$ is $\hat{\rho}$.
\item Construct a binary character $\alpha$ on $\hat{X}$ as follows: $$\alpha(x):=\begin{cases} 1 &\mbox{if }  x \in C\\ 0 & \mbox{if } x \in \hat{X}\setminus C. \end{cases}$$ Note that $\alpha(z)=0$ as $z \not\in X$ and thus $z \not\in C.$ 
\item Calculate the parsimony score $PS_\text{sw}(\hat{N},\alpha)$ using algorithm $\mathcal{A}$. 
\end{enumerate} 

Note that $PS_\text{sw}(\hat{N},\alpha)=1$ if and only if $N$ displays a tree $T$ which has a subtree with label set $C$. This is due to the fact that, as $\alpha(z)=0$, the softwired parsimony score of $\hat{N}$ can only be 1 if $\rho$ and $\hat{\rho}$ receive state 0. Otherwise, there would be a change required on one of the edges $(\hat{\rho},z)$ or $(\hat{\rho}, \rho)$ and additionally at least one more change in the part of $\hat{N}$ corresponding to $N$, as $X$ employs both states 1 and 0 for taxa in or not in $C$, respectively, because $C \subsetneqq X$. Moreover, if $\rho$ is in state 0, the softwired parsimony score of $\hat{N}$ is 1 precisely if $\hat{N}$ displays a tree $T$ which only requires one change, and that change has to be a change from 0 to 1. This is the case if and only if $N$ displays a tree $T$ with a subtree with leaf labels $C$. This case is illustrated by Figure~\ref{fig:extension}. Note that if $\mathcal{A}$ is polynomial, so is $\mathcal{\tilde{A}}$. Therefore, computing the softwired parsimony score of a binary character on a binary rooted phylogenetic network is NP-hard.
\end{proof}

We can extend the NP-hardness result to a more restricted class of rooted phylogenetic networks.

\begin{theorem}\label{softwired_treechild_thm}
Computing the softwired parsimony score of a binary tree-child time-consistent rooted phylogenetic network and a binary character is NP-hard.
\end{theorem}
\begin{proof}
We can make any network tree-child and time-consistent by hanging a cherry \cs{in the middle of} each reticulation edge. If we give the two leaves of each cherry character states~0 and~1, then the softwired parsimony score is increased exactly by the number of added cherries.
\end{proof}

\begin{corollary}\label{cor:notfpt}
It is NP-hard to decide if the softwired parsimony score of a binary rooted time-consistent phylogenetic network and a binary character is equal to one. In particular, there is no fixed-parameter tractable algorithm with the parsimony score as parameter unless~P~=~NP.
\end{corollary}
\begin{proof}
\cs{It has been proven by \cite{ISS2010b} that \textsc{Cluster Containment} is NP-hard even for time-consistent networks by reducing from \textsc{Cluster Containment} on general networks. Since the proof by  \cite{ISS2010b} transforms the input network in a way that preserves binarity, the corollary follows directly from the combination of the result by \cite{ISS2010b} and Theorem~\ref{thm:rootedhard}.}
\end{proof}

\begin{figure}[H]
    \centering
    \includegraphics[scale=0.7]{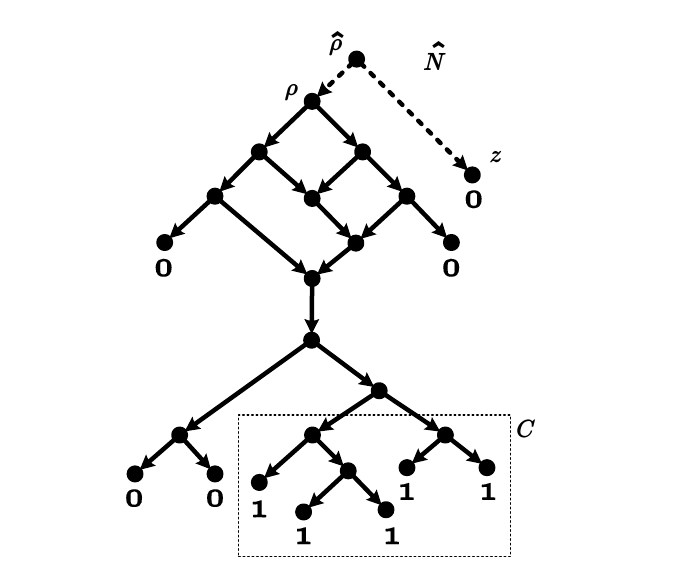} 
    \caption{Illustration of the extension of a rooted binary phylogenetic network $N$ (solid lines) to the rooted binary phylogenetic network $\hat{N}$ as described in the proof of Theorem~\ref{thm:rootedhard}. The additional taxon $z$ is assigned state 0 along with all taxa in $X\setminus C$, whereas all taxa in $C$ are assigned state 1. Then, the softwired parsimony score of $\hat{N}$ is 1 if and only if $N$ displays a tree $T$ with a pending subtree with leaf set $C$.
\label{fig:extension}}
\end{figure}

\subsection{Complexity of Approximating the Softwired Parsimony Score}\label{subsec:approx}

\leo{In order to proceed to the question of approximability, we require a new definition and a lemma.} Recall that
$\cS(N)$ is the set of all switchings of a network. Given a rooted phylogenetic network $N$, we define $PS_{\cS}(N,\alpha)$ as:
\[
\min_{S\in\cS(N)} \min_{\tau} \sum_{e\in E(S)} c_\tau(e)
\]
where the second minimum is taken over all extensions~$\tau$ of~$\alpha$ to~$V(S)$.

The following \leo{lemma states} that optimal solutions can
equivalently be modelled as selecting the lowest-score switching, ranging over all
extensions $\tau$ of a character $\alpha$ to the nodes of the \emph{network}. This is the characterisation of optimality used in Section~\ref{ilp_sec} and enables us to circumvent some of the \snew{suppression and deletion} technicalities associated with the concept ``display''. \leo{Since the lemma is intuitively clear, we defer its proof to the appendix.}

\begin{lemma}
\label{lem:equivalenceDef}
Consider a rooted phylogenetic network~$N$ on~$X$ and a $p$-state character~$\alpha$ on~$X$. Then 
\[
PS_{\cS}(N,\alpha) = PS_\text{sw}(N,\alpha).
\]
\end{lemma}

The following straightforward corollary will be useful when describing approximation-preserving reductions.

\begin{corollary}
\label{cor:tree2labelling}
Given a network $N$ and a character $\alpha$ on $X$, a tree $T \in \cT(N)$, a switching $S \in \cS(N)$ corresponding to $T$ and an extension $\tau$ of $\alpha$ to $V(T)$, we can construct in polynomial time an extension $\tau'$ of $\alpha$ to $V(N)$ such that $\sum_{e \in E(S)} c_{\tau'}(e)$ $\snew{=} \sum_{e \in E(T)} c_{\tau}(e)$.
\end{corollary}

We now show that it is very hard to approximate the softwired parsimony score on rooted networks. \sbsteven{We give two inapproximability results. The first, the stronger of the two, applies to nonbinary networks and
holds even when the network is both tree-child and time-consistent. It shows that in a complexity-theoretic sense trivial
approximation algorithms are the best one can hope for in this case. The second result, which is only slightly weaker, applies to binary networks.} Both results are much stronger than the APX-hardness result presented by \cite{jin2009parsimony}.
At the present time we do not have an inapproximability result for networks that are simultaneously binary, tree-child (and time-consistent): in this sense
Theorem \ref{softwired_treechild_thm} is \snew{currently the strongest hardness result we have} for such networks.

Before proceeding we formally define the output of an algorithm that
approximates $PS_\text{sw}(N,\alpha)$ as a tree $T \in \cT(N)$ and a certificate
that $T \in \cT(N)$ i.e. a switching $S \in \cS(N)$ corresponding to $T$. The certificate
is necessary because it is NP-hard to determine whether a tree is displayed by a network
\cite{clustercontainment}. The parsimony score (i.e. value of the objective function) associated with \snew{the output} $T$ is then
\[
PS(T,\alpha) = \min_{\tau} \sum_{e\in E(T)} c_\tau(e)
\]
where the minimum is taken over all extensions~$\tau$ of~$\alpha$ to~$V(T)$. Note
that $PS(T,\alpha)$ can easily be computed in polynomial time by applying Fitch's algorithm to $T$. If necessary Corollary~\ref{cor:tree2labelling} can then be applied to transform this in polynomial time into an extension of $\alpha$ to $V(N)$ \cs{such that the} switching corresponding to $T$, i.e. our certificate, has parsimony score
at most $PS(T,\alpha)$.

\sbsteven{Consider the following simple observation}.

\begin{observation}
\label{obs:trivial}
The softwired parsimony score of a rooted phylogenetic network $N$ on $X$ and a $p$-state character $\alpha$ on  $X$   can be
(trivially) approximated in polynomial time with approximation factor $|X|$, for any~$p\geq 2$.
\end{observation}
\begin{proof}
\steven{
Let $s \in \{1,\ldots ,p\}$ be the state to which at least a fraction $1/p$ of $X$ is mapped by $\alpha$. Let~$T$ be an arbitrary tree in $\cT(N)$. We extend $\alpha$ to $V(T)$ by labelling
all internal nodes of $T$ with~$s$. Clearly, $PS_{sw}(N,\alpha)=0$ if and only if $\alpha$ maps all elements in $X$ to the same character state, in which case the extension of $\alpha$ to $V(T)$ also yields a parsimony score of~0. Otherwise, $PS_{sw}(N,\alpha) \geq 1$ and the extension described yields a parsimony score of at most $(1-1/p)|X|<|X|$, from which the result follows.}
\end{proof}

The following theorem shows that, in an asymptotic sense, Observation \ref{obs:trivial} is actually
the best result possible, even when the topology of the network is quite heavily restricted.

\begin{theorem}
\label{thm:sat}
For every constant $\epsilon > 0$ there is no polynomial-time approximation algorithm that can
approximate $PS_{sw}(N,\alpha)$ to a factor $|X|^{1-\epsilon}$, where $N$ is a \steven{tree-child, time-consistent} network and $\alpha$ is a binary
character on $X$, unless~P~=~NP.
\end{theorem}
\begin{proof}
We reduce from the NP-hard decision problem 3-SAT. This is the problem of determining whether a boolean formula
in CNF form, where each clause contains at most 3 literals, is satisfiable. Let $B = (V,C)$ be an instance of 3-SAT,
where $V$ is the set of variables and $C$ is the set of clauses. Let $|V|=n$. Observe that $|C|=m$ is at most
$O(n^3)$ because in a decision problem it makes no sense to include repeated clauses.

For each constant $\epsilon > 0$, we will show how to construct a parsimony instance $(N,\alpha)$ such that the existence of a polynomial-time $|X|^{1-\epsilon}$ approximation would allow us to determine in polynomial time whether $B$ is
a YES or a NO instance, from which the theorem will follow. \cyan{The construction can be thought of as an ``inapproximability'' variant of the hardness construction used by \cite{clustercontainment}}.

Throughout the proof we will make heavy use of the equivalence described in
Lemma~\ref{lem:equivalenceDef}. Specifically, we will characterise optimal solutions
to the softwired parsimony problem as the score yielded by the lowest-score switching, ranging over all extensions of $\alpha$ to $V(N)$.

We begin by proving the result for networks that are time-consistent, but not tree-child. Later we will show how to extend the result to networks that are time-consistent and
tree-child.

The centrepiece of the construction is the following \emph{variable gadget}. Let $z$ be a variable in $V$. We
introduce two nodes which we refer to as $z$ and $\neg z$, and name them collectively \emph{connector} nodes. We introduce two sets of taxa, $X_{z,0}$ and
$X_{z,1}$, each containing  $f(n,\epsilon)$ taxa, where $f(n,\epsilon)$ is a function that we will specify later.
For each taxon $x \in X_{z,i}$ we set $\alpha(x)=i$. By introducing $2 \cdot f(n,\epsilon)$ reticulation nodes
we connect  each taxon in $X_{z,0}$ and $X_{z,1}$ to both $z$ and $\neg z$ (see Figure \ref{fig:sat}). Observe that if
both $z$ and $\neg z$ are labelled with the same character state, the parsimony score of this gadget (and thus
of the network as a whole) will be at least $f(n,\epsilon)$. On the other hand, if $z$ and $\neg z$ are labelled
with different character states, the gadget contributes (locally) zero to the parsimony score. The idea is thus
that we label $(z, \neg z)$ with $(1,0)$ if we wish to set variable $z$ to be TRUE, and $(0,1)$ if we wish $z$
to be FALSE i.e. $\neg z$ is TRUE. By choosing $f(n,\epsilon)$ to be very large we will ensure that $z$ and $\neg z$
are never labelled with the same character state \snew{in ``good'' solutions}.

We construct one variable gadget for each $z \in V$. Next we add the root $\rho$ and two nodes $s_0$ and $s_1$. We
connect $\rho$ to $s_0$ and to $s_1$. Next, we connect $s_0$ (respectively, $s_1$) to every connector node (ranging over all
variable gadgets).  Hence, every connector node has indegree 2. The idea is that (without loss of generality) $s_0$ (respectively, $s_1)$ can be
assumed to be labelled 0 (respectively, 1). Therefore, if a connector node is labelled with state 0 (respectively, 1), 
it will choose $s_0$ (respectively, $s_1)$ to be its parent, and these edges will not contribute any mutations to the parsimony score. There are two
points to note here. Firstly, there
will be exactly one mutation incurred on the two edges $(\rho, s_0)$ and $(\rho, s_1)$, and this has an important role in the ensuing inapproximability
argument; we shall return to this later. Secondly, it could happen that the labelling of $s_0$ and $s_1$ is $(1,0)$ rather than $(0,1)$, but \snew{in that case the analysis is entirely symmetrical}.

It remains only to describe the \emph{clause gadgets}. These are very simple. For each clause $c \in C$ we introduce a size $f(n,\epsilon)$ set of taxa that
we call $X_{c}$. For each taxon $x \in X_c$ we set $\alpha(x)=1$. By introducing $f(n,\epsilon)$ nodes - \snew{these will be reticulations, unless the clause contains only one literal} - we connect each taxon in $X_{c}$ to the connector nodes in the variable
gadgets corresponding to the literals in the clause. For example, if $c$ is the clause $(\neg x \vee y \vee z)$, each node in $X_c$ has
$\neg x$, $y$ and $z$ as its parents. Similarly to the variable gadgets, observe that if none of the literals corresponding to $c$ are labelled 1 (i.e. set
to TRUE), the clause gadget corresponding to $c$ will raise the parsimony score by at least $f(n,\epsilon)$, but if that at least one literal is TRUE,
the (local) parsimony cost will be zero.


\begin{figure}
\centering
\includegraphics[scale=.6]{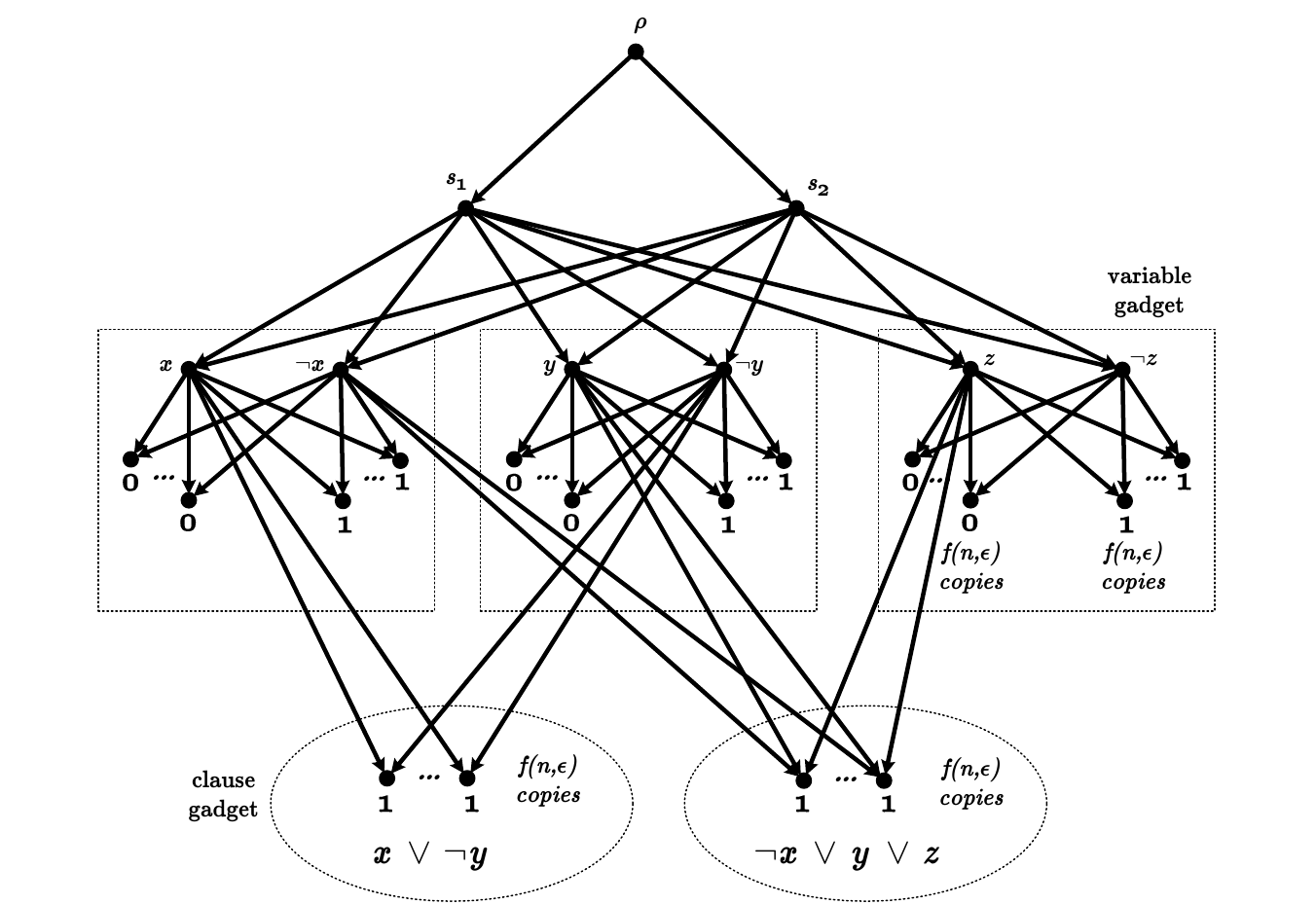}
\caption{An encoding of the 3-SAT instance $(x \vee \neg y) \wedge (\neg x \vee y \vee z)$
as described in Theorem \ref{thm:sat}. Note that the network is time-consistent -- a possible time-stamp allocates 1 to the root, 2 to all reticulation nodes plus $s_0$ and $s_1$, and 3 to the nodes in any single-literal
clause gadgets, as these are tree nodes -- 
but not tree-child: a slight modification is required to make it tree-child.
\label{fig:sat}}
\end{figure}

Observe firstly that $PS_{sw}(N,\alpha) \geq 1$ because both character states appear in the range of $\alpha$. More fundamentally, $PN_{sw}(N, \alpha)=1$ if
$B$ is satisfiable - in which case the single mutation occurs on one of the edges $(\rho, s_0)$ and $(\rho, s_1)$ - and $PN_{sw}(N, \alpha) \geq f(n, \epsilon)$ if $B$ is unsatisfiable. This dichotomy holds
because, in order to have $PS_{sw}(N, \alpha)$ $< f(n, \epsilon)$, it is necessary that the connector nodes in the variable gadgets always have \steven{a labelling of} the form $(0,1)$ or $(1,0)$, and
that for every clause $c$ the nodes in $X_c$ all have at least one TRUE parent i.e. $B$ is satisfiable.

The high-level idea is to choose $f(n,\epsilon)$ to be so large that even a weak approximation factor will be sufficient to determine without error whether $B$ is satisfiable or unsatisfiable. 
\steven{Before explaining how to choose $f(n,\epsilon)$ we formally describe the \snew{steps in the} reduction. Let
$T \in \cT(N)$ be the tree produced by the approximation algorithm for $PS_{sw}(N, \alpha)$,
and $S \in \cS(N)$ a corresponding switching. We compute $PS(T,\alpha)$ and let $\tau$ be any extension of $\alpha$ to $V(T)$ that achieves parsimony
score $PS(T,\alpha)$; this can all be done in polynomial time using Fitch's algorithm. If $PS(T,\alpha) \geq f(n,\epsilon)$ we declare that the SAT instance $B$ is unsatisfiable. Otherwise, we declare that $B$ is satisfiable. Note that, by Corollary \ref{cor:tree2labelling}, we can use $T, S$ and $\tau$ to obtain in polynomial time an
extension $\tau'$ of $\alpha$ to $V(N)$ such that the parsimony score of $S$ under $\tau'$
is also strictly less than $f(n,\epsilon)$. Therefore, for each variable $z$ in the SAT
instance, $\tau'$ has to label $z$ and $\neg z$ with different character states, and
for each clause $c \in C$ in the SAT instance, at least one of its literals has to be labelled
with character state 1. The satisfying assignment is thus: for each variable $z$, $z$ is TRUE if $z$ is labelled 1, and FALSE if $\neg z$ is labelled 1. 
 }

We now show how to choose $f(n,\epsilon)$ such that $|X|^{1-\epsilon} \cdot 1 < f(n,\epsilon)$. When chosen this way, an approximation algorithm with approximation factor $|X|^{1-\epsilon}$ will be forced
to return a solution that, \steven{as we have just described, can be transformed into} a satisfying assignment of $B$, whenever that is possible. This will be the only option, because returning a solution with parsimony score $f(n,\epsilon)$ or higher will be more than a factor $|X|^{1-\epsilon}$ larger
than the optimum, which is 1 in the case of satisfiability. Now, observe that $|X| = 2n \cdot f(n,\epsilon) + m \cdot f(n,\epsilon)$. Given the relationship between $n$ and $m$, a (crude) upper bound on $|X|$ is
$n^5 \cdot f(n,\epsilon)$, for sufficiently large $n$. Hence it is sufficient to ensure $f(n,\epsilon)^{1-\epsilon} \cdot n^{5(1-\epsilon)} < f(n,\epsilon)$. Suppose $f(n, \epsilon) = n^{g(\epsilon)}$, where $g(\epsilon)$ is a function that
only depends on $\epsilon$. Then we need $g(\epsilon)(1-\epsilon) + 5(1-\epsilon) < g(\epsilon)$, which implies that taking $g(\epsilon) = \lceil 6\epsilon^{-1}(1-\epsilon) \rceil$ is sufficient.

\steven{The network we constructed above is time-consistent (see Figure \ref{fig:sat})
 but
not tree-child: \snew{potentially} only the root $\rho$ has at least one child that is not a reticulation. We can
transform the network as follows. For each node $v$ with indegree greater than~1 and outdegree 0 we
simply add an outgoing edge to a new node $v'$, where $v'$ receives time-stamp 3, and
$\alpha(v')$ takes over the character state $\alpha(v)$. Next we introduce $2n+2$ new
taxa. For $s_i$, $i \in \{0,1\}$, we
introduce a new node $s'_i$ (with time-stamp 3), add an edge $(s_i, s'_i)$, and set $\alpha(s'_i) = i$. For each variable $z$ in the SAT instance $B$, we introduce two new \cs{taxa}
$z'$ and $\neg z'$ (both of which receive time-stamp 3), add edges $(z,z')$ and $(\neg z, \neg z')$ and set $\alpha(z')=\alpha(\neg z')=0$. The network is now both tree-child and
time-consistent. Now, observe that the two taxa introduced underneath $s_0$ and $s_1$
do not change the optimum parsimony score, because without loss of generality we can
assume that
$s_0$ is labelled 0 and $s_1$ is labelled 1. However, for
each variable $z$, some extra mutations might be incurred on the edges
$(z,z')$ and $(\neg z, \neg z')$. As long as $f(n,\epsilon)$ is chosen to be large enough,
these (at most) $2n$ extra mutations do not significantly alter the reduction: in optimal
solutions each $(z, \neg z)$ pair will still be labelled with different character states, reflecting
a satisfying truth assignment for $B$, whenever $B$ is satisfiable. In fact, if $B$ is
satisfiable, then \cs{exactly} $n$ extra mutations will be incurred on the edges
$(z,z')$ and $(\neg z, \neg z')$ (in an optimal solution), since at least one of $z$ and $\neg z$ will be labelled
0. So, if $B$ is satisfiable, $PS_{sw}(N, \alpha) \cs{=} n+1$ and if $B$ is unsatisfiable,
$PS_{sw}(N, \alpha) \geq f(n,\epsilon)$. As long as we choose
$|X|^{1-\epsilon}(n+1) < f(n,\epsilon)$, any $|X|^{1-\epsilon}$-approximation will
be forced to return a tree $T$ such that $PS(T,\alpha) < f(n,\epsilon)$ whenever $B$ is satisfiable,
and this can be transformed in the same way as before in polynomial time into a satisfying assignment for $B$.}

Hence we need to choose $f(n,\epsilon)$ such that
$|X|^{1-\epsilon}(n+1) < f(n,\epsilon)$, where this time 
$|X| = 2n \cdot f(n,\epsilon) + m \cdot f(n,\epsilon) + (2n+2)$. As before,
$|X| \leq n^5 \cdot f(n,\epsilon)$ holds for sufficiently large $n$, as does $n+1 < n^2$. So establishing
$f(n,\epsilon) > n^{(7-5\epsilon)}f(n,\epsilon)^{1-\epsilon}$ would
be sufficient. Letting $f(n,\epsilon)=n^{g(\epsilon)}$ and taking logarithms, it
is sufficient to choose $g(\epsilon)$ such that $g(\epsilon) > 7 - 5\epsilon + g(\epsilon)(1-\epsilon)$. Taking $g(\epsilon) = \lceil \frac{7-5\epsilon}{\epsilon} \rceil +1$ is
sufficient for this purpose, and we are done.
\end{proof}


For binary networks, we get a slightly weaker inapproximability result.

\begin{theorem}
\label{thm:sat2}
For every constant~$\epsilon > 0$ there is no polynomial-time approximation algorithm that approximates $PS_{sw}(N,\alpha)$ to a factor $|X|^{\frac{1}{3}-\epsilon}$, where~$N$ is a rooted binary phylogenetic network on~$X$ and $\alpha$ is a binary character on~$X$, unless~P~=~NP.
\end{theorem}
\begin{proof}
We reduce again from 3-SAT. Let, as before, $B=(C,V)$ be an instance of 3-SAT. Let $V=\{v_1,\ldots ,v_n\}$ and $F:=f(n,\epsilon)$. We will describe a construction of a rooted binary phylogenetic network~$N$ and binary character~$\alpha$. The first part of the construction is essentially a binary version of the network constructed in the proof of Theorem~\ref{thm:sat}. The main difference will be the construction of the so-called ``zero-gadgets''. For ease of notation, we will create vertices with indegree-1 and outdegree-1, which could be suppressed, and we create reticulations with indegree greater than~2, which could be refined arbitrarily. Observe that neither suppressing indegree-1 and outdegree-1 vertices nor refining reticulations with indegree greater than~2 alters the softwired parsimony score. However, to simplify the proof we do not suppress or refine these vertices. Furthermore, we assume without loss of generality that each literal is contained in at least one clause.

Our construction is as follows. We create a root~$\rho$ with two directed paths $(\rho,a_1,\ldots ,a_{2n})$ and $(\rho,b_1,\ldots ,b_{2n})$ leaving it. Then, for each variable~$v_i$, we create a reticulation vertex which we will also call~$v_i$ and has reticulation edges $(a_{2n-2i+2},v_i)$, $(b_{2i-1},v_i)$. Moreover, we create a reticulation vertex~$\neg v_i$ with reticulation edges $(a_{2n-2i+1},\neg v_i),(b_{2i},\neg v_i)$. The vertices~$v_i$ and~$\neg v_i$ are called ``literal vertices''. Then, for each clause~$c$, create~$F$ vertices $c_1,\ldots ,c_F$, which we will call ``clause vertices'', and for each such clause vertex~$c_f$, create an edge $(c_f,c'_f)$ to a new leaf $c'_f$ with character state $\alpha(c'_f)=1$ (the ``clause leaves''). So, in total we have~$|C|F$ clause leaves.

We now connect the literal vertices to the clause vertices. For each variable~$x$ and its negation~$\neg x$, we do the following. Suppose that~$x$ is in~$T$ clauses, corresponding to~$F\cdot T$ clause vertices, $c_x^1,\ldots ,c_x^{F\cdot T}$. Create a directed path $(x,d_x^1,\ldots ,d_x^{F\cdot T})$ and edges $(d_x^k,c_x^k)$ for~$k=1,\ldots ,{F\cdot T}$. Similarly, if~$\neg x$ is in~$\Theta$ clauses with~$F\cdot\Theta$ corresponding clause vertices $(\gamma_x^1,\ldots ,\gamma_x^{F\cdot\Theta})$, we create a directed path $(\neg x,\delta_x^1,\ldots ,\delta_x^{F\cdot\Theta})$ and edges $(\delta_x^\kappa,\gamma_x^\kappa)$ for~$\kappa=1,\ldots ,{F\cdot\Theta}$. Now, for each $k\in\{1,\ldots ,F\cdot T\}$ and for each $\kappa\in\{1,\ldots ,F\cdot\Theta\}$, we create the following ``zero-gadget''. Let~$u$ be the parent of~$c_x^k$ that is reachable from~$x$ (hence, in the first iteration, $u=d_x^k$). Replace the edge~$(u,c_x^k)$ by a directed path $(u,u_1,\ldots ,u_F,c_x^k)$. Similarly, let~$\mu$ be the parent of~$\gamma_x^\kappa$ that is reachable from~$\neg x$ (initially, $\mu=\delta_x^\kappa$), and replace the edge $(\mu,\gamma_x^\kappa)$ by a directed path $\mu,\mu_1,\ldots ,\mu_F,\gamma_x^\kappa$. Then, for~$f=1,\ldots ,F$, create a new reticulation~$z_f$, with reticulation edges~$(u_f,z_f)$ and~$(\mu_f,z_f)$, and an edge $(z_f,z'_f)$ to a new leaf~$z'_f$ with character state $\alpha(z'_f)=0$.

See Figure~\ref{fig:binarySAT} for an example of the construction for~$F=1$. For larger~$F$, the construction is similar but with more copies of each clause vertex, and more copies of each zero-gadget.

The constructed network~$N$ has $|C|F$ clause leaves (all having character state~1) and at most~$|V||C|^2F^3$ leaves for the zero-gadgets (all having character state~0). Hence, the total number of leaves is at most $|C|F + |V||C|^2F^3$.

\begin{figure}
\centering
\includegraphics[scale=.6]{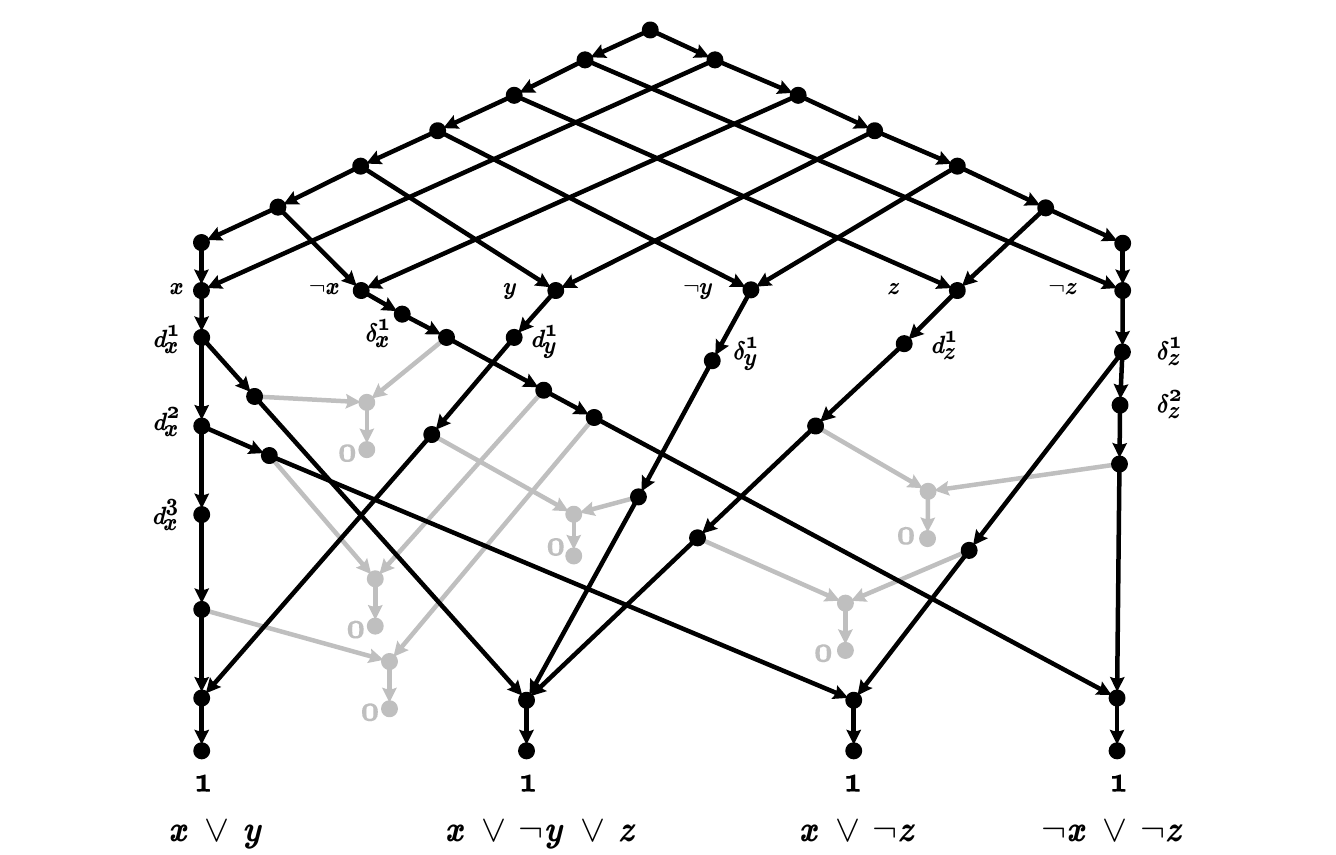}
\caption{An encoding of the 3-SAT instance $(x \vee y) \wedge (x \vee \neg y \vee z)\wedge (x\vee \neg z) \wedge (\neg x\vee\neg z)$ as described in Theorem \ref{thm:sat2}, with~$F=1$. The zero-gadgets are indicated in grey.
\label{fig:binarySAT}}
\end{figure}

We will show that if~$(C,V)$ is satisfiable, $PS(N_\text{sw},\alpha)=1$ and that if~$(C,V)$ is not satisfiable, $PS(N_\text{sw},\alpha)\geq F$. We will use the following definitions. For two vertices~$u$ and~$v$, we say that~$v$ is a \emph{tree-descendant} of~$u$ if~$v$ is reachable from~$u$ by a directed path that does not contain any reticulations apart from possibly~$u$. In particular, each vertex is a tree-descendant of itself. Furthermore, given an extension~$\tau$ of~$\alpha$ to~$V(N)$, we say that there is a \emph{change at} vertex~$v$ if the character state of~$v$ is different from the character states of all its parents. The \emph{number of changes} of network~$N$ and extension~$\tau$ is the number of vertices at which there is a change. It follows from Lemma~\ref{lem:equivalenceDef} that the softwired parsimony score of a network~$N$ and character~$\alpha$ is equal to the minimum number of changes over all possible extensions~$\tau$ of~$\alpha$ to~$V(N)$.

First suppose that~$(C,V)$ is satisfiable. Then, given a satisfying truth assignment, we can assign character states as follows. All vertices on the path $(\rho,a_1,\ldots ,a_{2n})$ and all clause vertices~$c_f$ receive state~1. All vertices on the path $(b_1,\ldots ,b_{2n})$ and all reticulations~$z_f$ of zero-gadgets receive state~0. For each variable~$x$ that is set to true by the truth assignment, we give state~1 to all tree-descendants of literal vertex~$x$ and state~0 to all tree-descendants of literal vertex~$\neg x$. Similarly, for each variable~$x$ that is set to false by the truth assignment, we give state~0 to all tree-descendants of literal vertex~$x$ and state~1 to all tree-descendants of literal vertex~$\neg x$. This concludes the assignment of character states. Now observe the following. Consider a clause~$c$ and a corresponding clause vertex~$c_f$, which has state~1. Since~$c$ is satisfied by the truth assignment, at least one parent of~$c_f$ has also state~1. Hence, there are no changes at the clause vertices. Moreover, for each reticulation~$z_f$ of a zero-gadget (which has state~0), there is at least one parent that also has state~0 because for each variable~$x$ either all tree-descendants of~$x$ or all tree-descendants of~$\neg x$ have state~0. Using these observations, it can easily be checked that the only change is at~$b_1$. Hence, if~$(C,V)$ is satisfiable, $PS(N_\text{sw},\alpha)=1$.

Next, we show that if~$(C,V)$ is not satisfiable, $PS(N_\text{sw},\alpha)\geq F$. We do this by assuming that $PS(N_\text{sw},\alpha) < F$ and showing that this implies that $(C,V)$ is satisfiable. Let~$\tau$ be an extension of~$\alpha$ to~$V(N)$ with less than~$F$ changes. For a positive literal~$x$ and clause vertex~$c_x^k$ for a clause containing~$x$, let $P(x,c_x^k)$ denote the directed path from~$d_x^k$ to~$c_x^k$ (with $d_x^k$ as defined in the construction of~$N$). Similarly, for a negative literal~$\neg x$ and clause vertex~$\gamma_x^\kappa$ for a clause containing~$\neg x$, let $P(\neg x,\gamma_x^\kappa)$ denote the directed path from~$\delta_x^\kappa$ to~$\gamma_x^\kappa$ (with $\delta_x^\kappa$ as defined in the construction of~$N$). Moreover, for any literal~$\ell$ (of the form~$x$ or~$\neg x$) and clause vertex~$c_f$ for a clause containing~$\ell$, let~$P'(\ell,c_f)$ denote path~$P(\ell,c_f)$ excluding its first vertex. We compute a truth assignment as follows. A variable~$x$ is set to true if and only if for some clause vertex~$c_x^k$ for a clause containing~$x$ holds that all vertices on the path $P(x,c_x^k)$ have state~1. We now prove that the obtained truth assignment is a satisfying truth assignment. Assume that a certain clause~$c$ is not satisfied. Consider a clause vertex~$c_f$ corresponding to clause~$c$. Observe that, for two different clause vertices~$c_{f_1},c_{f_2}$ corresponding to clause~$c$ and for any two literals~$\ell_1,\ell_2$ contained in clause~$c$, the paths $P'(\ell_1,c_{f_1}),P'(\ell_2,c_{f_2})$ are vertex-disjoint. Hence, since~$\tau$ has less than~$F$ changes, there exists at least one clause vertex~$c_f$ corresponding to clause~$c$ for which there are no changes at~$c'_f$ or at any vertex on a directed path $P'(\ell,c_f)$ for any literal~$\ell$ contained in clause~$c$. Since~$c'_f$ has state~1, it follows that~$c_f$ has state~1, and hence that at least one parent of~$c_f$ has state~1, and hence that there exists at least one literal~$\ell$ contained in clause~$c$ such that all vertices on the path $P(\ell,c_f)$ have state~1. If~$\ell$ is of the form~$x$ (a positive literal), then this immediately implies that~$\ell$ is set to true contradicting the assumption that clause~$c$ is not satisfied. Now consider the case that~$\ell$ is of the form~$\neg x$ (a negative literal). Let~$\kappa$ be such that $\gamma_x^\kappa=c_f$. We have shown that all vertices on the path~$P(\ell,c_f)$ from~$\delta_x^\kappa$ to~$\gamma_x^\kappa=c_f$ have state~1. For every $k\in\{1,\ldots ,F\cdot T\}$, there is a zero-gadget for~$x$,~$k$ and~$\kappa$. Each such zero-gadget contains a directed path $(u_1,\ldots ,u_F)$ on~$P(x,c_x^k)$ and a directed path $\mu_1,\ldots ,\mu_F$ on~$P(\neg x,\gamma_x^\kappa)$. Since all vertices on the path $P(\neg x,\gamma_x\kappa)$ have state~1, and there are less than~$F$ changes, at least one vertex of the path $(u_1,\ldots ,u_F)$ has state~0. Hence, at least one vertex on $P(x,c_x^k)$ has state~0 for all~$k\in\{1,\ldots ,F\cdot T\}$. It follows that~$x$ is set to false and hence that literal~$\ell=\neg x$ is set to true, contradicting the assumption that clause~$c$ is not satisfied. Therefore, we have shown that, if~$(C,V)$ is not satisfiable, $PS(N_\text{sw},\alpha)\geq F$.

It remains to describe how to choose $F=f(n,\epsilon)$ such that $|X|^{\frac{1}{3}-\epsilon}<f(n,\epsilon)$. Recall that $|X|\leq |C|f(n,\epsilon) + |V||C|^2f(n,\epsilon)^3$. Then, with~$n=|V|$ and recalling that~$|C|=O(n^3)$, we can bound this by $|X| \leq n^8 f(n,\epsilon)^3$ for sufficiently large~$n$. Hence, it is enough to show that $n^{8(\frac{1}{3}-\epsilon)} f(n,\epsilon)^{3(\frac{1}{3}-\epsilon)} < f(n,\epsilon)$. Taking $f(n,\epsilon) = n^{g(\epsilon)}$, we need $8(\frac{1}{3}-\epsilon)+3(\frac{1}{3}-\epsilon)g(\epsilon) < g(\epsilon)$. Hence, it is sufficient to take $g(\epsilon) = \lceil\frac{8}{9\epsilon}\rceil$.
\end{proof}

In particular, Theorem~\ref{thm:sat2} shows that there can be no~$O(\log(|X|))$-approximation for computing the softwired parsimony score of a binary rooted phylogenetic network, unless~P~=~NP. We remark that the network constructed in the proof of Theorem~\ref{thm:sat2} can not easily be made tree-child. Hence the inapproximability of binary tree-child networks is still open. It does seem that the constructed network can be made time-consistent but we omit a proof.

Although we have shown above that there is no algorithm for computing the softwired parsimony score that is fixed-parameter tractable in the parsimony score (unless~P~=~NP), there obviously exists such an algorithm that is fixed-parameter tractable in the reticulation number of the network: \snew{a network with reticulation number $r$}
has $2^r$ switchings, and for each switching Fitch's algorithm can be used. Moreover, in the next section we show that there even exists an algorithm that is fixed-parameter tractable in the level of the network, \snew{a parameter potentially much smaller than reticulation number}.

\section{An FPT Algorithm in the Level of the Network for Computing the Softwired Parsimony Score of a Network}\label{sec:FPT}

In the first part of this section we describe a polynomial-time dynamic programming \snew{(DP)} algorithm that works
on rooted trees and computes a slight generalisation of the softwired parsimony score.  We then show how this can be used as a subroutine in computing
the softwired parsimony score of networks, such that the running time is fixed-parameter tractable (FPT) in the level of
the network.

\subsection{A DP Algorithm for (not Necessarily Phylogenetic) Rooted Trees with Weights \label{subsec:treeDP}}

\cs{Let $\mathcal{P} = \{1, \ldots, p \}$} be the set of character states. Let $T$ be a rooted tree, and let $L(T)$ be the set of leaves of $T$.  $T$ is not necessarily a phylogenetic tree because only a subset $L \subseteq L(T)$ need to be \cs{labelled}, and $T$ is allowed to have \snew{nodes} with indegree and outdegree both equal to 1. (Later on, we will see that this allows us to model switchings). For a \snew{node} $v \in V(T)$, let $T_{v}$ be the subtree of $T$ rooted at $v$.

We are given a $p$-state character $\alpha : L \rightarrow \mathcal{P}$. Additionally, we are given a function
$w : (V(T) \times \mathcal{P}) \rightarrow \mathbb{N}$, where $\mathbb{N} = \{0,1,\ldots\}$.

\snew{Consider the following definition, where $c_{\tau}$ is the \emph{change} function described in the preliminaries
and the minimum ranges over all extensions $\tau$ of $\alpha$ to $V(T)$.}
\begin{equation}
\label{eq:def}
PS_{sw}(T, \alpha, w) = \min_{\tau} \bigg ( \bigg ( \sum_{e \in E(T)} c_{\tau}(e) \bigg ) + \bigg ( \sum_{v \in V(T)} w(v,\tau(v)) \bigg ) \bigg ).
\end{equation}
We can think of this as being the parsimony score with an optional added ``weighting'' that to varying degrees \snew{``penalises"} \snew{nodes} when they are allocated a certain character state. This weighting $w(v,s)$ will be used in the next section to model the contribution to the optimum parsimony score of the subnetworks of $N$ rooted at $v$ when $v$ is forced to be labelled with character state $s$.
 
To compute $PS_{sw}(T, \alpha,w)$ we introduce the value
$PS_{sw}(T, \alpha, w,s)$, with $s \in \mathcal{P}$, where we add the restriction that the root
of $T$ must be labelled with character state $s$. Clearly,
\begin{equation}
\label{eq:minrange}
PS_{sw}(T, \alpha,w) = \min_{s \in \mathcal{P}} PS_{sw}(T, \alpha, w,s)
\end{equation}
We denote by $PS_{sw}(T, \alpha,w, \cdot)$ the vector $(PS_{sw}(T, \alpha,w, 1), \cdots, \snew{PS}_{sw}(T, \alpha,w, p))$. 
This vector is computed as described in Algorithm \ref{alg:PSw}, where $\delta(s,s') = 0$ if
$s = s'$ and 1 otherwise, and $C(v)$ is the set of children of a non-leaf \snew{node} $v$.  \snew{Note that
the optimal $\tau$ can be constructed by backtracking, if necessary}.

\begin{algorithm}[H]
\small{
\For{each \snew{node} $v$ of $V(T)$ considered in post-order}{
\If{$v$ is a leaf}{
\If{$v \in L$}{ $PS_{sw}(T_v, \alpha, w,s) = w(v,s)$, if $s=\alpha(v)$ and $\infty$ otherwise;}
\Else{$PS_{sw}(T_v, \alpha,w,s) = w(v,s)$, for each $s \in \mathcal{P}$;}
}
\Else{
$PS_{sw}(T_v, \alpha,w,s) = w(v,s) + \displaystyle\sum_{v' \in C(v)} \bigg ( \min_{s' \in \mathcal{P}} \bigg ( PS_{sw}(T_{v'}, \alpha, w,s') + \delta(s,s') \bigg ) \bigg )$, $\forall s \in \mathcal{P}$;
}
}
\Return  $PS_{sw}(T, \alpha, w,\cdot)$; \texttt{//note that $T=T_{root(T)}$}
}
\caption{Compute $PS_{sw}(T, \alpha,w,\cdot)$ \label{alg:PSw}}
\end{algorithm}

\medskip
The running time of Algorithm~\ref{alg:PSw} is $O(p^2|V(T)|)$.


\snew{
\begin{lemma}
\label{lem:alg1correct}
Algorithm \ref{alg:PSw} correctly computes $PS_{sw}(T_v, \alpha,w,\cdot)$, for every $v \in V(T)$. In particular,
it correctly computes $PS_{sw}(T,\alpha,w,\cdot)$.
\end{lemma}
\begin{proof}
\emph{(Sketch)} This follows from the fact that, if the state of a node $v$ is fixed as $s$, then the only 
local decisions that have to be made
to optimize $PS_{sw}(T_v, \alpha,w,s)$ are to choose the character state~$s'$ for each child~$v'$ of~$v$. A change is incurred whenever $s'\neq s$.
Once $s'$ has been chosen, we are free to (and therefore should) use optimal subsolutions corresponding
to the case when the root of subtree $T_{v'}$ has state $s'$ i.e. $PS_{sw}(T_{v'}, \alpha,w,s')$. We omit details.
\end{proof}
}

The following lemma shows that Algorithm~\ref{alg:PSw} can be used to compute the parsimony score of a phylogenetic tree.

\begin{lemma}
Consider a rooted tree $T$ on $X$ and a $p$-state character $\alpha$ on $X$. Then, if 
$w(v,s)=0$ for all $v \in V(T)$ and $s \in \mathcal{P}$, then $PS_{sw}(T, \alpha,w)=PS_{sw}(T, \alpha)$.
\label{lem:DP_trees}
\end{lemma}
\begin{proof}
\snew{This follows by combining Lemma \ref{lem:alg1correct} with (\ref{eq:minrange}) and (\ref{eq:def}). In particular,
in (\ref{eq:def}) the right-hand side of the expression degenerates to the familiar parsimony definition, because $w$ is 0 everywhere.}
\end{proof}

\subsection{Extending the DP Algorithm to Networks}
\label{subsec:extendDP}

Let $N$ be a level-$k$ network on $X$. We say that a biconnected component is \emph{trivial} if it consists of a single edge (a cut-edge). Thanks to the \snew{following} results, we can envisage $N$ as comprising non-trivial biconnected components, each with reticulation number at most $k$, arranged in a tree-like backbone. 

\begin{lemma}
Let $N$ be a rooted phylogenetic network on $X$ and let $B$ be a biconnected component of $N$. Then $B$ contains exactly one node $r_B$ without ancestors in $B$.
\label{lem:proofOneRoot}
\end{lemma}
\begin{proof}
Suppose there exist two roots in $B$, $r_1$ and $r_2$. In a rooted network $N$ there always exists a directed path from the root of $N$ to each node in $N$. Hence there exists a node $v$ in $N$ such that there is a simple directed
path from $v$ to $r_1$, and a simple directed path from $v$ to $r_2$. (Note that we do not exclude the possibility
that $v \in \{r_1, r_2\}$.) By merging these two paths we see that there is an undirected
simple path $P$ between $r_1$ and $r_2$ such that for at least one of $r_1$ and $r_2$ the edge of $P$ incident to it is oriented towards it. We want to argue that all nodes and edges of $P$ are also in $B$, which will contradict the assumption that $r_1$ and $r_2$ are both roots of $B$. In fact, it holds that if any two nodes $u$ and $v$ in $B$ have a simple undirected path $P$ between them, all nodes and edges of $P$ are also in $B$. If this was not true, then $P$ would contain some node
not in $B$, and this in turn would mean that, in the journey from $u$ to $v$, $P$ would have to pass through some cut \snew{node} twice, contradicting
its simplicity. Hence all nodes of $P$ are in $B$, and by maximality all of the edges of $P$ are too.
\end{proof}

\begin{lemma}
Let $N$ be a rooted phylogenetic network on $X$. Then, if $r$ is a reticulation, all incoming edges of $r$ are in the same biconnected component of $N$.\label{lem:proofOneComponent}
\end{lemma}
\begin{proof}
The lemma can be proven by applying a similar \snew{argument to that used} in the proof of Lemma \ref{lem:proofOneRoot}. We therefore omit the proof. 
\end{proof}

We define the switchings of a biconnected component~$B$ of~$N$ analogous to the definition of switchings of a network, i.e. a \emph{switching} of a biconnected component~$B$ is a rooted tree~$S_B$ that can be obtained from~$B$ by deleting all but one of the incoming edges of each reticulation. We say that we \emph{apply} \snew{switching} $S_B$ to $N$ when deleting in $N$ all edges of $B$ not in $S_B$. The next result is a consequence of Lemma \ref{lem:proofOneComponent}.

\begin{lemma}
Let $N$ be a rooted phylogenetic network on $X$ and $S$ a switching of $N$. Then $S$ can be obtained from $N$ by, for each biconnected component $B$ of $N$, first choosing a switching $S_B$ and then applying it to $N$.
\label{lem:partialSwitchings}
\end{lemma}

\begin{corollary}
Let $N$ be a rooted phylogenetic network on $X$, $S$ a switching of $N$ and $B$ a biconnected component of $N$. Let $S_B$ be the switching of $B$ \snew{induced by $S$} and let $S'_B$ be a different switching of $B$. Let $S'$ be the graph obtained from $N$ by applying to $N$ 
all switchings of all biconnected components \snew{induced by $S$} except $S_B$, \snew{and then finally applying the switching $S'_B$}. Then $S'$ is a switching of $N$.
\label{cor:partialSwitchings}
\end{corollary}

We are now ready to describe our algorithm for computing the softwired parsimony score of a phylogenetic network~$N$ and $p$-state character~$\alpha$. Note that each cut-edge $(u,v)$ is seen as a biconnected component with root~$u$ and only one switching.

\begin{algorithm}[]
\lFor{each \snew{node} $v$ of $N$ and state $s$ in $\mathcal{P}$}{ $w(v,s) \gets 0$\;}
\For{each \snew{node} $r$ of $N$  that is a root of a least one biconnected component in post-order}{
\For{each biconnected component $B_{r}$ rooted at $r$}{
\For{each switching $S_{r}$ of $B_{r}$}{
compute $PS_{sw}(S_{r}, \alpha, w, \cdot)$ using Algorithm \ref{alg:PSw}\;
}
\For{each $s$ in $\mathcal{P}$}{
$PS_{sw}(B_{r}, \alpha,w,s) = \displaystyle\min_{S_{r} \in \mathcal{B}_{r}} PS_{sw}(S_{r}, \alpha, w,s)$\;
$w(r,s) \gets w(r,s) + PS_{sw}(B_{r}, \alpha, w,s)$\;
}
}
}
\Return $\displaystyle\min_{s\in \mathcal{P}} w(root(N),s)$\;
\caption{
\snew{Compute} $PS_{sw}(N,\alpha)$}\label{algo:FPT}

\end{algorithm}

\begin{theorem} \label{softwired_levelFPT_thm}
Computing the softwired parsimony score of a rooted phylogenetic network $N$ and a $p$-state character, for any~$p\in\NN$, is fixed-parameter tractable if the parameter is the level of the network.
\end{theorem}
\begin{proof}
In Lemma \ref{lem:proofOneRoot}, we proved that each biconnected component $B$ of $N$ contains only one root $r_B$. We denote by $BT(N)$ the graph obtained as follows: we create a node $v_r$ in $BT(N)$ for each \snew{node}~$r$ of~$N$ that is the root of at least one biconnected component of $N$, and an edge $(v_r,v_{r'})$ in $BT(N)$ if $r'\neq r$ and~$r'$ is contained in a biconnected component rooted at~$r$. It is easy to see that $BT(N)$ is connected. Moreover, it cannot contain any reticulation because of Lemma \ref{lem:proofOneComponent}. Thus $BT(N)$ is a tree on $X$. 
In the following we shall prove that $PS_{sw}(N,\alpha)= \min_{s\in \mathcal{P}} w(root(N),s)$ {with~$w$ the weight function computed by Algorithm~\ref{algo:FPT}}. 

Denote by $PS_{sw}(N, \alpha,s)$ the minimum parsimony score for $N$ and $\alpha$, with the restriction that the root
of $N$ must be labelled with character state $s$. \cs{Let  $N_r$ be the subnetwork comprising all biconnected components whose roots can be reached by
directed paths from $r$}. We will prove that  $PS_{sw}(N_r, \alpha,s)=w(r,s)$ for any \snew{node} $r$  in $V(N)$ associated to a \snew{node} $v_r$ in $BT(N)$. We prove this equality by induction on the height of~$v_r$, which is defined as the length of a longest path from~$v_r$ to a leaf of~$BT(N)$.

We begin by proving that the equality is true when the height of~$v_r$ is 0. Suppose that $r$ is the root of $J$ different non trivial biconnected components and let $B_j$ be one of these. Then we have that:

\noindent 
$PS_{sw}(B_j,\alpha,w,s)=\displaystyle\min_{S_{j} \in \mathcal{B}_{j} } PS_{sw}(S_{j}, \alpha, w,s)=\min_{S_{j}  \in \mathcal{B}_{j} } PS_{sw}(S_{j}, \alpha,s)=PS_{sw}(B_j,\alpha,s)$,

 \noindent where the second equivalence holds because of Lemma \ref{lem:DP_trees}, since $w(v,s)$ is equal to zero for all nodes of $B_j$ and $s$ in $\mathcal{P}$. 
 Then, because of Lemma \ref{lem:proofOneComponent}, we have that:
 
  $w(r,s) = \sum_{B_j \in \mathcal{B}_r} PS_{sw}(B_j,\alpha,w,s) = \sum_{B_j \in \mathcal{B}_r} PS_{sw}(B_j,\alpha,s)=PS_{sw}(N_r,\alpha,s) $,
 
\noindent where $\mathcal{B}_r$ is the set of biconnected components rooted at $r$.

Suppose now that $w(r,s)=PS_{sw}(N_r,\alpha,s) $ is true for all nodes $v_r$ of $BT(N)$ with height at most~$h$. We want to prove that this holds also for nodes with height $h+1$. Let $v_r$ be such a node, $r$ the associated node in $N$ and let $B_j$ a biconnected component rooted at $r$. 
Let  $N_r(B_j)$ denote the subnetwork of $N_r$ where edges not reachable from $B_j$ are deleted. 
Then, when we loop through all switchings of $B_j$, we can use the same kind of reasoning as before to prove that  $PS_{sw}(B_j,\alpha,w,s)=PS_{sw}(N_r(B_j),\alpha,s)$. The idea is that, once a subnetwork has been processed, its influence in the biconnected component above it is expressed using the $w$ function. This implies that the claim holds, since $PS_{sw}(N_r,\alpha,s) = \sum_{B_j \in \mathcal{B}_r} PS_{sw}(N_r(B_j),\alpha,s)=\sum_{B_j \in \mathcal{B}_r} PS_{sw}(B_j,\alpha,w,s)=w(r,s)$. 
 
We still need to prove the running time. Algorithm \ref{alg:PSw} has a running time of $O(p^2|V(T)|)$. Moreover, for each biconnected component we call Algorithm~\ref{alg:PSw} for at most $2^k$ trees with at most $|V(N)|$ nodes. 
Moreover, by Lemma \ref{lem:proofOneComponent}, we have that the number of 
biconnected components of $N$ is at most $|E(N)|$. Then we have an overall complexity of $O(2^kp^2 |V(N)| \cdot |E(N)|)$. This concludes the proof.  
\end{proof}
 
\section{Maximum Parsimony in Practice: Integer Linear Programming}\label{ilp_sec}

We propose the following integer linear programming (ILP) formulation for computing the hardwired parsimony score of a phylogenetic network, with \snew{node} set~$V$ and edge set~$E$, and a $p$-state character~$\alpha$. All variables are binary. Variable~$x_{v,s}$ indicates whether or not \snew{node}~$v$ has character state~$s$ and variable~$c_e$ indicates if there is a change on edge~$e$ or not. For a leaf~$v$, parameter~$\alpha(v)$ is the given character state at~$v$. Let~$\cP=\{1,\ldots ,p\}$.
\begin{align*}
\min & \sum_{e\in E} c_e &\\
\text{s.t. }
& \sum_{s\in \cP} x_{v,s} = 1							&&\text{for all } v\in V\\
& c_e \geq x_{u,s} - x_{v,s} 						&&\text{for all } e=(u,v)\in E, s\in \cP \\
& c_e \geq x_{v,s} - x_{u,s} 						&&\text{for all } e=(u,v)\in E, s\in \cP \\
& x_{v,\alpha(v)}=1 								&&\text{for each leaf } v \\
& c_e\in\{0,1\} 									&&\text{for all } e\in E \\         
& x_{v,s}\in\{0,1\}     								&&\text{for all } v\in V, s\in \cP 
\end{align*}
To see the correctness of the formulation, first observe that the first constraint ensures that each \snew{node} is assigned exactly one character state. Now consider an edge~$e=(u,v)$ and suppose that~$u$ and~$v$ are assigned different states~$s$ and~$s'$. Then $x_{u,s} \neq x_{v,s}$ (and $x_{u,s'} \neq x_{v,s'}$) and hence the second and third constraint ensure that $c_e=1$.

For the softwired parsimony score, we extend the ILP formulation as follows. In addition to the variables above, there is a binary variable~$y_e$ indicating if edge~$e$ is switched ``on'' or ``off''. A change on edge~$e$ is only counted if it is switched on. For each reticulation, exactly one incoming edge is switched on.
\begin{align*}
\min & \sum_{e\in E} c_e &\\
\text{s.t. }
& \sum_{s\in \cP} x_{v,s} = 1							&&\text{for all } v\in V\\
& c_e \geq x_{u,s} - x_{v,s} - (1-y_e)				&&\text{for all } e=(u,v)\in E, s\in \cP \\
& c_e \geq x_{v,s} - x_{u,s} - (1-y_e)				&&\text{for all } e=(u,v)\in E, s\in \cP \\
& \sum_{v:(v,r)\in E} y_{(v,r)} = 1 					&&\text{for each reticulation } r \\
& y_e=1 											&&\text{for each tree-edge } e \\
& x_{v,\alpha(v)}=1 								&&\text{for each leaf } v \\
& c_e,y_e\in\{0,1\} 								&&\text{for all } e\in E \\         
& x_{v,s}\in\{0,1\}     								&&\text{for all } v\in V, s\in \cP 
\end{align*}
It follows from Lemma~\ref{lem:equivalenceDef} that the optimum value of this ILP is equal to the softwired parsimony score of the given network and character.

Note that the parsimony score of an alignment can be computed by solving the above ILP formulation for each column (character) separately, or combining them in a single ILP. Gaps in the alignment can be accommodated in the formulation by demanding that $x_{v,\alpha(v)}=1$ only for leaves $v$ for which $\alpha(v)$ is not a gap.

We have implemented both ILP \leo{formulations} and made the resulting user-friendly software publicly available~\cite{MPNet}. Experimental results with CPLEX 12.5 on a 2Ghz laptop are in Table~\ref{tab:exp}. Networks were simulated using Dendroscope~\cite{Dendroscope3} and character-states were assigned uniformly at random.

\begin{table}
\begin{tabular}{c|c|c|c|c|c|c|c|c|}
\cline{2-9}
&   & avg. &\multicolumn{6}{|c|}{Average computation time (s)} \\
& $|\mathcal{X}|$ & number of & \multicolumn{3}{|c|}{Hardwired PS} & \multicolumn{3}{|c|}{Softwired PS}\\
& & retic. & 2-state & 3-state & 4-state & 2-state & 3-state & 4-state\\
\hline
Run 1 & 50 & 17.0 & 0.0 & 0.0 & 0.1 & 0.1 & 0.1 & 0.3\\
Run 2 & 100 & 37.0 & 0.0 & 0.0 & 0.2 & 0.0 & 0.1 & 0.6\\
Run 3 & 150 & 54.1 & 0.0 & 0.1 & 0.6 & 0.1 & 0.2 & 0.8\\
Run 4 & 200 & 72.8 & 0.0 & 0.1 & 1.1 & 0.1 & 0.4 & 1.4\\
Run 5 & 250 & 91.3 & 0.0 & 0.1 & 3.5 & 0.1 & 0.4 & 2.2\\
Run 6 & 300 & 112.6 & 0.0 & 0.2 & 5.2 & 0.1 & 0.6 & 3.7\\
\hline
\end{tabular}
\caption{Time needed to compute hardwired and softwired parsimony scores in six test runs. For each run, an average is taken over 10 simulated networks with randomly assigned character states.\label{tab:exp}}
\end{table}

For practical applications, parsimony scores have to be computed quickly since this computation needs to be repeated many times, for example when searching for a network with smallest parsimony score. Apparent from Table~\ref{tab:exp} is that parsimony scores can be computed very quickly using ILP for networks with up to 100-150 taxa and up to 50 reticulations. Moreover, parsimony scores can even be computed quickly for much larger networks in the case of binary and ternary characters. This is of interest because, in practice, many columns of an alignment might contain only two or three different symbols. Despite the theoretical differences in tractability of hardwired and softwired parsimony scores, their computation times using ILP do not differ much in these experiments.

\section{Conclusions and Open Problems}\label{conclusion_sec}

\begin{table}
\begin{tabular}{|c||c|c|}
\hline
& Hardwired & Softwired\\
\hline \hline
\multirow{2}{*}{Complexity} & In P for $p=2$                                & \multirow{2}{*}{NP-hard for $p\geq 2$}\\
                                       & NP-hard for~$p\geq 3$                       & \\
\hline
\multirow{2}{*}{Approximation} & $\frac{12}{11}$-approx. for $p=3$ & \steven{no $|X|^{1-\epsilon}$-approx.}\\
                                           & 1.3438-approx. for~$p\geq 4$       & \steven{for any $\epsilon > 0$} unless ~P~=~NP\\
\hline
Parameterized by PS               & FPT                                           & NP-hard to decide if PS=1\\
\hline
Parameterized by level             & \cyan{N/A}                                           & FPT\\
\hline
\end{tabular}
\caption{Summary of the complexity of computing hardwired and softwired parsimony scores of phylogenetic networks.\label{tab:sum}}
\end{table}

We have clarified the distinction between two possible definitions of the parsimony score of a phylogenetic network, which we call the ``softwired" and the ``hardwired'' parsimony score. We have shown that computing the hardwired parsimony score is, in various ways, more tractable than computing the softwired score, see Table~\ref{tab:sum}. We have also shown that the intractability results still hold under several topological restrictions.
\sbsteven{A stimulating open} question is to determine the (in)approximability and fixed-parameter tractability of computing the softwired parsimony score of a character on a rooted network that is simultaneously binary and tree-child: \sbsteven{this might be
considerably more tractable than other versions of the problem.} From a practical point of view, we have shown that both the hardwired and softwired parsimony score can be computed efficiently using ILP. It will be interesting to explore -- in the spirit of studies such as \leo{those} conducted by 
\cite{MPcaseStudy2007} and \cite{jin2009parsimony} -- the extension of this work to the notoriously intractable ``big parsimony" problem.

\section*{Acknowledgements}
We thank Mike Steel for useful discussions on the topic of this paper. This work has been partially funded by the ANCESTROME project ANR-10-IABI-0-01. This publication is the contribution no. 2013-\textcolor{black}{XXX} of the Institut des Sciences de l'Evolution de Montpellier (ISE-M, UMR 5554). Leo van Iersel was funded by a Veni grant of the Netherlands Organisation for Scientific Research (NWO).

\bibliographystyle{plain}
  \bibliography{bibliographyleo}

\clearpage

\appendix

\section{Transforming degree-2 nodes (from Corollary \ref{hardwired_NP_cor})\label{app:transf}}

Although the hardwired parsimony score naturally extends to them, degree-2 nodes are not formally part of
our phylogenetic network model. Fortunately, degree-2 nodes can simply be suppressed without altering the hardwired
parsimony score. Unfortunately 
this may in turn create multi-edges which are likewise excluded from our definition. To deal with this, a multi-edge with multiplicity
$t \geq 2$ between two nodes $u$ and $v$ can be encoded within the degree restrictions of a phylogenetic network by using a specific gadget. Namely, group the edges into $t' = \lfloor t/2 \rfloor$ pairs and for each pair $P_i$ ($1 \leq i \leq t'$) (i) delete the
two edges concerned (ii) add two new nodes $x_i, y_i$ and (iii) add the edges $(u,x_i), (u,y_i), (x_i,y_i),(x_i,v),(y_i,v)$.
(If $t$ is odd the remaining edge can simply remain intact). Again, this does not
alter the hardwired parsimony score. In fact, both transformations also leave the cut properties of the graph unchanged,
which is important for the proof of Corollary \ref{hardwired_NP_cor}.

\section{Proof of Lemma~\ref{lem:equivalenceDef}}

\textbf{Lemma~\ref{lem:equivalenceDef}.}
\emph{Consider a rooted phylogenetic network~$N$ on~$X$ and a $p$-state character~$\alpha$ on~$X$. Then}
\[
PS_{\cS}(N,\alpha) = PS_\text{sw}(N,\alpha).
\]
\begin{proof}
Let $S$ be a switching of $N$ and  $\tau$  an extension of~$\alpha$ to~$V(S)$ -- or equivalently to $V(N)$ -- such that $\sum_{e\in E(S)} c_\tau(e)=PS_{\cS}(N,\alpha)$. Let $T$ be the tree obtained from $S$ by deleting indegree-0 outdegree-1 nodes, deleting unlabelled outdegree-0 nodes and suppressing indegree-1 outdegree-1 nodes and let $\tau'$ be the restriction of $\tau$ to the nodes of $S$ still present in $T$. By construction we have that, since $S$ is a switching of  $N$ on $X$, and $T$ has been obtained from $S$ as described above, then $T\in \cT(N)$. Moreover, since  $\tau$ is an extension of~$\alpha$ to~$V(S)$, we have that $\tau'$ is a an extension of~$\alpha$ to~$V(T)$. Finally, it is easy to see that $PS_{\cS}(N,\alpha) = \sum_{e\in E(S)} c_\tau(e) \geq \sum_{e\in E(T)} c_{\tau'}(e) \geq PS_\text{sw}(N,\alpha)$, since suppressing nodes   (and 	consequently edges) cannot increase the sum of changes on the remaining edges.   

Now, let  $T$ be a tree of $\cT(N)$ and  $\tau$  an extension of~$\alpha$ to~$V(T)$ such that $\sum_{e\in E(T)} c_\tau(e)=PS_\text{sw}(N,\alpha)$. Moreover, let $S$ be a switching corresponding to $T$, i.e. such that $T$ can be obtained from $S$ by deleting indegree-0 outdegree-1 nodes, deleting unlabelled outdegree-0 nodes and suppressing indegree-1 outdegree-1 nodes. 
(We know that such a switching exists  because  $T\in \cT(N)$). Now, let $\tau': V(S)\rightarrow \{1,...,p, ?\}$ such that $\tau'(u)=\tau(u)$ if $u\in V(T)$ \snew{(i.e. $u$ is the image in $N$ of a node of $T$)} and $\tau'(u)=\{?\}$ otherwise. 
A value in $\{1,...,p\}$ is associated to all nodes $u$ of $S$ having $\tau'(u)=\{?\}$ in the following way: We start by setting $\tau'(root(S))$ to $\tau(root(T))$ and then we traverse $S$ in preorder,  setting $\tau'(u)$ to $\tau'(u_p)$ for all nodes $u$ having $\tau'(u)=\{?\}$, where $u_p$ is the parent node of $u$. \\
\snew{
First note that the root of $T$ corresponds to the node $\rho$ of $S$ that is closest to the root with the following
property: $\rho$ has out-degree 2 or higher, and nodes not labelled $?$ can be reached from at least two children of $\rho$.}

\snew{Then we have that all edges of $S$ not reachable from $\rho$ cost 0}, since for all these edges $(u,v)$ we have  $\tau'(u)=\tau'(v)=\tau'(root(T))$. 
Now, let $e=(u,v)$ be an edge of $T$. In $S$ this edge will often correspond to a set of edges, denoted by $E_S(e)$, see Figure~\ref{fig:edgeEx}. Now, note that the value of  $\tau'(\cdot)$ is  equal to $\tau(u)$ for all  descendants of $u$ in $S$ that cannot be reached via $v$. Then it is easy to see that the cost of all edges in $E_S(e)$ equals $c_{\tau'}(w,v)=c_{\tau}(u,v)=c_\tau(e)$, where $w$ is the parent node of $v$ in $S$.	  

\begin{figure}[H]
    \centering
    \includegraphics[scale=0.6]{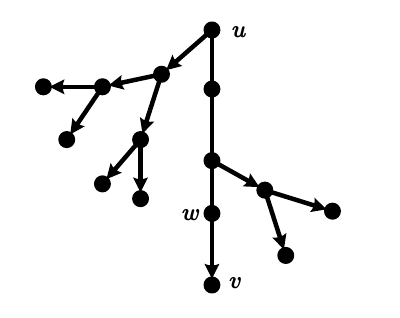}
    \caption{An example of an edge \snew{$(u,v)$} of $T$ that \snew{has been mapped to} several edges in the switching underlying $T$, used in the proof of Lemma \ref{lem:equivalenceDef}. \label{fig:edgeEx}}
\end{figure}

Since this holds for all edges of $T$, and $\tau'$ is clearly an extension of~$\alpha$ to~$V(S)$, we have that $PS_\text{sw}(N,\alpha)=  \sum_{e\in E(T)} c_{\tau}(e)   =  \sum_{e\in E(S)} c_{\tau'}(e) \geq PS_{\cS}(N,\alpha)$.
 This concludes the proof. 
\end{proof}

\end{document}